\documentclass[notitlepage,a4paper,aps,prd,onecolumn,superscriptaddress,nofootinbib,groupedaddress]{revtex4-2}

\usepackage{amssymb}
\usepackage{amsfonts}
\usepackage{amsmath}
\usepackage[ignoreall]{geometry}

\usepackage{empheq}

\usepackage{enumitem}

\newtheorem{theorem}{Theorem}

\newtheorem{corollary}[theorem]{Corollary}

\newtheorem{definition}[theorem]{Definition}

\newtheorem{proposition}[theorem]{Proposition}

\newenvironment{proof}[1][Proof]{\noindent\textbf{#1.} }{\ \rule{0.5em}{0.5em}}
\usepackage{hyperref}
\hypersetup{
    colorlinks=true,
    linkcolor=blue,
    filecolor=magenta,      
   citecolor=blue
}

\usepackage{color}

\definecolor{vert}{rgb}{0.,0.65,0.}

\newcommand{\diff}{\mathrm{d}}

\usepackage{mathrsfs}

\newcommand{\tens}[1]{\mathbf{#1}}
\newcommand{\tensi}[1]{\boldsymbol{#1}}
\newcommand{\metric}{\boldsymbol{g}}

\usepackage{accents}
\newcommand{\lc}[1]{\accentset{\circ}{#1}}

\newcommand{\ud}[3]{#1^{#2}_{\phantom{#2}#3}}
\newcommand{\du}[3]{#1^{\phantom{#2}#3}_{#2}}

\DeclareMathOperator{\met}{Met}
\DeclareMathOperator{\id}{Id}

\begin{document}

\title{Metric-affine cosmological models and the inverse problem of the
calculus of variations.\\
Part 1: variational bootstrapping - the method}

\author{Ludovic Ducobu}
\email{ludovic.ducobu@umons.ac.be}
\affiliation{Transilvania University of Brasov, Brasov, Romania}
\affiliation{University of Mons, Mons, Belgium}

\author{Nicoleta Voicu}
\email{nico.voicu@unitbv.ro}
\affiliation{Transilvania University of Brasov, Brasov, Romania}
\affiliation{Lepage Research Institute, Presov, Slovakia}

\begin{abstract}
The method of variational completion allows one to transform an (in principle, arbitrary) system of partial differential equations -- based on an intuitive ``educated guess'' -- into the Euler-Lagrange one attached to a Lagrangian, by adding a canonical correction term. Here, we extend this technique to theories that involve at least two sets of dynamical variables : we show that an educated guess of the field equations with respect to one of these sets of variables only is sufficient to variationally complete these equations and recover a Lagrangian for the full theory, up to terms that do not involve the respective variables.

Applying this idea to natural metric-affine theories of gravity, we prove that, starting from an educated guess of the metric equations only, one can find the full metric equations, together with a generally covariant Lagrangian for the theory, up to metric-independent terms; the latter terms (which can only involve the distortion of the connection) are then completely classified over 4-dimensional spacetimes, by techniques pertaining to differential invariants.
\end{abstract}

\maketitle




\section{Introduction}

General Relativity (GR) offers a model for gravity, based on Riemannian geometry, which is in excellent agreement with most observations but still exhibits problems at the largest (\emph{e.g.}, rotational curves of galaxies or the accelerated expansion of the Universe) and the smallest scales (\emph{e.g.}, tensions with quantum theory). This indicates that it is necessary to look for a more general gravitational theory. Whereas there is an ongoing debate on which extensions of Riemannian geometry 
are the most suitable, there is quite a wide consensus that the dynamics of such an extended model for gravity should be based on the principle of minimal action -- \emph{i.e.}, on the calculus of variations.

Typically, in building extensions of GR, once fixed the kinematics, one then starts by postulating a Lagrangian; the associated field equations and their consequences being subsequently derived from it. Examples of such modified theories include Horndeski gravity, where the spacetime geometry is still described by Riemannian geometry but an additional scalar degree of freedom is present, see \cite{horndeski, Kobayashi:2019hrl, Langlois:2018dxi}, and metric-affine theories of gravity, where the gravitational degree of freedom comes from both the spacetime metric and an independent affine connection, see \cite{Hehl:1976, Hehl:1994ue, Blagojevic:2012bc, CANTATA:2021ktz}.\\
Yet, the choice of these Lagrangians is very often based on formal arguments, aimed to \emph{indirectly} control the expected behaviour of the solutions of the associated Euler-Lagrange equations. While we would not argue against the soundness of this approach in itself, by first focusing on the Lagrangian, this sole procedure usually fails to single out a theory, or a class of theories, staying close from a desired phenomenological behaviour. One may then want to find a refined, or, at least, complementary, technique for selecting the best models.

On account for that, we propose to revert for a moment the roles and start from the field equations. More precisely, imagine we postulate a form of the desired \textit{field equations }(which may be variational or not), based on some physical principles or phenomenological considerations. A Lagrangian, together with the ``corrected'' -- fully variational -- field equations, can then be derived from this ``educated guess'', by the so-called technique of \textit{canonical variational completion} (or briefly, \textit{variational completion}), introduced\footnote{Whereas such a route had been taken in specific, isolated cases (\emph{e.g.}, Horndeski theories, \cite{horndeski}), a general and systematic procedure in this sense, was first introduced by D. Krupka together with the second author of the present paper, in  \cite{canonical-var-completion}.} in \cite{canonical-var-completion}.

To give an idea on this technique, a motivating example was the historically first version of the Einstein field equations:%
\begin{equation}
R_{\mu \nu }=\kappa T_{\mu \nu },  \label{incomplete_efe}
\end{equation}%
which accurately predicted some physical facts, but was inconsistent with local energy-momentum conservation. In this case, the variational completion procedure gives precisely%
\begin{equation}
R_{\mu \nu }-\frac{1}{2}Rg_{\mu \nu }=\kappa T_{\mu \nu }  \label{efe}
\end{equation}%
as ``corrected'' field equations; the canonical correction term (which is, in this example, $-\frac{1}{2} Rg_{\mu \nu }$) is expressed in terms of the \textit{Helmholtz form} of the given PDE system, \cite{Krupka-book}, which measures the obstruction from variationality of the original equations.

\bigskip

This paper is the first of a two-part work aiming to formulate metric-affine theories of gravity using a procedure based on variational completion.

In the present paper, we extend the variational completion method to situations where one wants to build a theory depending on more than one dynamical variable, but only has an ``educated guess'' at the field equations of \textit{one} of these dynamical variables, say $y^{A}$. In this case, we show in Theorem \ref{thm:partial_var_compl} that one can still completely determine a Lagrangian for the theory, up to boundary terms and terms that have no dependence on $y^{A}$ or its derivatives. Such a setup is particularly appealing for the construction of modified theories of gravity, in general, since once can apply this procedure with $y^{A}=g^{\mu \nu }$, taking lessons from GR for the ``educated guess'' of the metric field equations.

Here, once formally established the procedure, we show how it applies to metric-affine theories of gravity. Starting from some postulated (variational or not)\ metric equations $\mathfrak{G}_{\mu \nu}=0,$ we determine the closest variational equations to these. We then classify all possible corresponding generally covariant Lagrangians, up to boundary terms, as follows:
\begin{enumerate}
\item According to Theorem \ref{thm:partial_var_compl}, all the terms in the Lagrangian containing the metric or its derivatives can be ``bootstrapped'' by variational completion.
\item Terms that do not involve the metric (and thus cannot be recovered by the above procedure) are then shown, using specific techniques for finding differential invariants, \cite{Janyska}, \cite{Kolar}, to be of order at most one in the connection and polynomial (of bounded degree) in the distortion tensor components and their derivatives. A full list of these possibilities, over 4-dimensional  backgrounds, is presented in Appendix \ref{appx:classif}.
\end{enumerate}

In the forthcoming Part 2 of this work\footnote{currently under writing}, we apply the presented method to select the metric-affine models that produce equations closest to those of the $\Lambda$CDM model of cosmology. The present paper then aims to establish the formal aspects underlying the application presented in Part 2.

\bigskip

The paper is structured as follows : Section \ref{sec:varcompl} is a somewhat minimalist (though, we hope, self-contained) review of the method of variational completion. Section \ref{sec:bootstrap} then extends variational completion to the situation where one only has an educated guess at the form of a part of the equations. We also present one example showing how the procedure behaves in a well-known case. Further, in Section \ref{sec:metric-affine}, we apply this method to the case of metric-affine theories, to find all possible metric-dependent terms in the Lagrangian, starting from an approximate form (educated guess) of the metric equations. All metric-independent generally covariant Lagrangians on 4-dimensional spacetimes are then determined in Section \ref{sec:torsion}. Finally, in Section \ref{sec:Conclusion}, we summarise our findings and present further directions of research.

\section{Variational completion of differential equations}
\label{sec:varcompl}

This section presents in brief the method of canonical variational completion introduced in \cite{canonical-var-completion}. For more details on the formalism of the calculus of variations on fibered manifolds, we refer to the monographs \cite{Krupka-book} and \cite{Giachetta}.

Consider an arbitrary PDE system of $n$ equations of order $r$, in $n$ dependent variables, say $y^{B} = y^{B}\left(x^{\mu}\right)$ : 
\begin{equation}
\mathcal{E}_{A}(x^{\mu },\partial_\mu y^B, \cdots, \partial_{\mu_1}\cdots\partial_{\mu_r}y^B)=0.  \label{PDE}
\end{equation}%

The inverse problem of the calculus of variations consists in finding out if the given system is (locally or globally) \textit{variational}, \emph{i.e.}, if it arises (locally, respectively, globally) as the Euler-Lagrange system attached to some Lagrangian. The question of local variationality is usually answered by means of some differential relations involving $\mathcal{E}_{A},$ called the  \textit{Helmholtz conditions}. In the event of a negative answer, it is sometimes (under some quite strong constraints) possible to transform the given system into an equivalent one which is locally variational, by means of \textit{variational multipliers}, see \cite{inv-problem-book}.

The method of variational completion proposes a completely different approach: it transforms the given PDE system (\ref{PDE}) into a -- generally, non-equivalent -- variational one, by canonically adding a correction term. The canonical correction term, which is expressed in terms of the Helmholtz coefficients of the given system, is built so as to measure the obstructions from variationality of the original system and can, with a\ few exceptions, always be constructed. The method can thus also act as an elegant way of checking variationality.

\subsection{Geometric setup : fields, Lagrangians, Euler-Lagrange forms and source forms}

In Lagrangian field theories, physical fields are understood as local sections of fibered manifolds.

\bigskip

A fibered manifold is a triple $(Y,\pi ,M),$ where $M$ and $Y$ are smooth manifolds of dimensions $m$ and $m+n$ respectively, and $\pi :Y\rightarrow M$ is a surjective submersion of class $\mathcal{C}^{\infty}$, called the projection. $Y$ is called the \textit{total space} and $M$ the \textit{base manifold}. 
This structure allows one to find on $Y$ an atlas consisting of \textit{fibered charts} $\left(V,\psi \right)$, $\psi =\left( x^{\mu }, y^{A}\right)$, in which $\pi$ is represented as
\begin{equation}
\pi :\left( x^{\mu },y^{A}\right) \mapsto \left( x^{\mu }\right),
\end{equation}
$\left(\pi(V), \psi_M\right)$, $\psi_M \coloneqq (x^\mu)$, being a chart on the base manifold $M$.
Fibered manifolds are physically interpreted as configuration spaces of physical systems and include, as a subclass, fiber bundles.

Most typically, these configuration spaces actually belong to the more restrictive subclass of \textit{natural bundles}, which are obtained via a ``universal'', functorial construction $\mathfrak{F} : M\mapsto Y \coloneqq \mathfrak{F}M$ that applies to the whole category of $m$-dimensional manifolds (\emph{e.g.}, the tangent/cotangent bundle, tensor bundles, connection bundles). On natural bundles, each local coordinate system $\left( x^{\mu }\right) $ on $M$ induces so-called natural coordinates $\left( x^{\mu },y^{A}\right) $ on $Y;$ accordingly, each coordinate change on $M$ induces a natural coordinate change on $Y$ (\emph{e.g.}, by a tensor, or connection coefficients rule).

In the following, unless elsewhere specified, by $\left( Y,\pi ,M\right)$ we will denote a general fibered manifold (with no extra structure).

\bigskip

Local \textit{sections }$\gamma :U\rightarrow Y$, defined by $\pi \circ \gamma =\id_{U}$ (where $U\subset M$ is open) are described in fibered coordinates by $(x^{\mu }) \mapsto \left( x^{\mu },y^{A}\left( x^{\mu }\right) \right)
,$ \emph{i.e.}, by specifying a dependence betweeen the (\textit{a priori} independent) coordinates $x^{\mu }$ and $y^{A}$ :
\begin{equation}
y^{A}=y^{A}\left( x^{\mu }\right).
\end{equation}%
These are, as announced above, physically interpreted as \textit{fields}.

\bigskip

Adding derivatives of the field variables (up to some order $r$) into the picture is tantamount to considering jets of order $r$ of sections : given a section $\gamma : U\to Y$ and $x\in U$,
\begin{equation}
\label{eq:jet}
J_{x}^{r}\gamma \coloneqq \left( x^{\mu };y^{A}\left( x^{\mu }\right) ;\partial_{\nu }y^{A}\left( x^{\mu }\right) ;\cdots;\partial _{\nu _{1}}\cdots\partial_{\nu _{r}}y^{A}\left( x^{\mu }\right) \right) .
\end{equation}%
It is worth mentioning that the definition of the jet is independent of the choice of fibered charts, see \cite{Krupka-book}.

The set $J^{r}Y=\left\{ J_{x}^{r}\gamma ~|~\gamma :U\rightarrow Y~\text{-section, }x\in U\right\} $ of jets of all local sections of $Y,$ at all points of their domains, is called the \textit{jet bundle} of order $r$ of $Y$. This is a smooth manifold, with naturally induced coordinate charts $\left( V^{r},\psi ^{r}\right)$, with
\begin{equation}
V^{r}\coloneqq J^{r}V,~\ \ \ \psi ^{r}\coloneqq (x^{\mu},y^{A},y_{~\nu }^{A},\cdots,y_{~\nu _{1}\cdots\nu _{r}}^{A}).
\end{equation}
The coordinate functions $y_{~\nu _{1}\cdots\nu _{s}}^{A}$ (where $\nu _{1}\leq\cdots\leq \nu _{s}$ and $s\leq r$) can be interpreted as slots into which, inserting a \textit{prolonged section} $J^{r}\gamma : U\rightarrow J^{r}Y : x\mapsto J_{x}^{r}\gamma ,$ one gets the partial derivatives up to order $r$ at $x$ of the functions $y^{A}=y^{A}\left( x^{\mu }\right) $ associated with $\gamma$; that is
\begin{equation}\label{eq:placeholder}
y_{~\nu _{1}\cdots\nu _{s}}^{A}\left( J_{x}^{r}\gamma \right) =\partial _{\nu_{1}}\cdots\partial _{\nu _{s}}y^{A}\left( x^{\mu }\right) .
\end{equation}

Jet bundles $J^rY$ come naturally equipped with the structure of a fibered manifold (induced by that of $Y$) in multiple ways :
\begin{enumerate}
\item With base $M$, with projection
\begin{equation}
\label{eq:pir}
\pi^r : J^rY \to M : J^r_x\gamma \mapsto \pi^r(J^r_x\gamma) \coloneqq x.
\end{equation}

\item With base $J^sY$, for $s\le r$, with projection
\begin{equation}
\label{eq:pir}
\pi^{r,s} : J^rY \to J^sY : J^r_x\gamma \mapsto \pi^{r,s}(J^r_x\gamma) \coloneqq J^s_x\gamma,
\end{equation}
which ``forgets all derivatives of order more than $s$''.
\end{enumerate}

\bigskip

In principle, a locally defined differential form on $J^{r}Y$ is expressed, in fibered charts, as a linear combination of wedge products of the differentials $\diff x^{\mu }$, $\diff y^{A},$ $\diff y_{~\mu}^{A},\cdots,\diff y_{~\mu_{1}\cdots\mu _{r}}^{A}.$ Yet, for our purposes, only two classes of such forms will be relevant:

\begin{enumerate}
\item \textit{Horizontal }$k$-\textit{forms}, with $k\leq m$, which can be written (in one, then, in any fibered chart) as linear combinations of wedge products $\diff x^{\mu _{1}}\wedge \cdots\wedge \diff x^{\mu _{k}}$ only. Their set will be denoted by $\Omega _{k,M}\left( V^{r}\right)$.

The most important subclass is represented by \textit{Lagrangians }of order $r$, which are defined as horizontal forms of maximal rank $k=m.$ In coordinates : 
\begin{equation}
\lambda =\mathcal{L}\diff x\in \Omega _{m,M}\left( V^{r}\right) ,
\end{equation}%
where $\diff x\coloneqq \diff x^{0}\wedge \cdots\wedge \diff x^{m-1}$ and $\mathcal{L}=\mathcal{L}(x^{\mu},y^{A},y_{~\nu }^{A},\cdots,y_{~\nu _{1}\cdots\nu _{r}}^{A})$ is the Lagrangian density. 

In particular, on a natural bundle $\mathfrak{F}M $, a \textit{generally covariant, }or \textit{natural}, Lagrangian of order $r$ is a globally defined Lagrangian $\lambda \in \Omega_{m,M}\left( J^{r}\mathfrak{F}M\right) $ which is invariant under changes of natural coordinates on $J^{r}\mathfrak{F}M $ induced by completely arbitrary coordinate changes (local diffeomorphisms) of $M$; in other words, a natural Lagrangian is a Lagrangian that makes sense globally, with the same formula, over any $m$-dimensional spacetime manifold $M$.

\item \textit{Source forms, }or \textit{dynamical forms, }which are locally defined $\left( m+1\right) $-forms expressible (in any fibered chart) as :
\begin{equation}
\mathcal{E=E}_{A}\diff y^{A}\wedge \diff x\in \Omega _{m+1}(V^{r}),
\label{source_form}
\end{equation}%
where $\mathcal{E}_{A}=\mathcal{E}_{A}(x^{\mu},y^{A},y_{~\nu }^{A},\cdots,y_{~\nu_{1}\cdots\nu _{r}}^{A})$. For a coordinate--free definition of the concept, employing the notion of contact form and bundle projections, we refer to \cite{Krupka-book}. The most notorious example of a source form is the \textit{Euler-Lagrange form }$\mathcal{E}\left( \lambda \right) $ of a Lagrangian $\lambda$, given by
\begin{equation}
\mathcal{E}_{A}=\dfrac{\delta \mathcal{L}}{\delta y^{A}}.
\end{equation}%
More generally, any PDE system of $m=\dim M$ equations of order $r,$ in the unknowns $y^{A}=y^{A}\left(x^\mu\right)$ can be encoded into a source form; for instance, equations (\ref{PDE}) above can be encoded as :
\begin{equation}
J^{r}\gamma^* \mathcal{E} = 0,
\end{equation}
that is, using a fibered chart,
\begin{equation}\label{eq:sourceformonchart}
\mathcal{E}_{A}\circ J^{r}\gamma = 0,
\end{equation}
where $\gamma $ is the section given by $y^{A}=y^{A}\left(x^\mu\right)$.
\end{enumerate}

A source form $\mathcal{E}$ on $J^{r}Y$ is called \textit{locally variational} if, for any chart domain $V^{r}\subset J^{r}Y,$ there exists a Lagrangian $\lambda _{V}$ whose Euler-Lagrange form is $\mathcal{E}$. Accordingly, $\mathcal{E}$ is called \textit{globally variational} if there exists a Lagrangian $\lambda$ defined throughout $J^{r}Y,$ such that $\mathcal{E} = \mathcal{E}\left(\lambda \right)$.

\bigskip

Before going further, we should introduce one operation allowing to build horizontal forms on jet bundles. The horizontalization operator is the unique mapping $h:\Omega(J^{r}Y)\rightarrow \Omega(J^{r+1}Y)$, compatible with the wedge product\footnote{that is, given $\theta, \rho$ differential forms on $J^rY$, the differential form (on $J^{r+1}Y$) $h(\theta\wedge\rho) = h\theta\wedge h\rho$.} (\emph{i.e.} $h$ is a morphism of exterior algebras) and obeying, in any fibered chart,
\begin{equation}
hf=f\circ \pi ^{r+1,r} \text{ and } h\diff f=\diff_{\mu}f\diff x^{\mu},  \label{eq:hdf}
\end{equation}%
for all $f : J^{r}Y \to\mathbb{R}$; here, $\diff_{\mu}f \coloneqq \partial _{\mu}f + \dfrac{\partial f}{\partial y^{A}}y_{~\mu}^{A}+ \cdots  + \dfrac{\partial f}{\partial y_{~\nu_{1}\cdots \nu_{r}}^{A}}y_{~\nu_{1}\cdots \nu_{r}\mu}^{A}$ denotes
the \textit{total derivative} (of order $r+1$) with respect to $x^{\mu}.$
As a direct consequence of \eqref{eq:hdf}, one has
\begin{equation}
h\diff x^{\mu} = \diff x^{\mu},\ h\diff y^{A} = y_{~\mu}^{A}\diff x^{\mu}, \cdots,\ h\diff y_{~\nu_{1}\cdots \nu_{k}}^{A } = y_{~\nu_{1}\cdots \nu_{k}\mu}^{A}\diff x^{\mu},\ ~\ \ k=1,\cdots ,r.  \label{eq:horizontalization_basis}
\end{equation}
Relations \eqref{eq:horizontalization_basis} and the compatibility with the wedge product then ensure that, given a form $\rho\in\Omega_{k}(J^rY)$, $h\rho\in\Omega_{k,M}(J^{r+1}Y)$; \emph{i.e.} $h\rho$ is horizontal.

Actually, the latter relations point out a natural differential operator
\begin{equation*}
\diff_H \coloneqq h \circ \diff : \Omega(J^{r}Y)\rightarrow \Omega(J^{r+1}Y),
\end{equation*}	
where $\diff : \Omega(J^rY) \to \Omega(J^rY)$ is the exterior derivative on $J^rY$. This operator is called the horizontal (or total) exterior derivative, \cite{Anderson}, and obeys :
\begin{equation} \label{eq:d_H^2}
\diff_H \circ \diff_H \equiv 0,
\end{equation}
(the latter relation can be easily checked as total derivatives commute).

By construction, $\diff_H$ relates to what remains of the exterior derivative on $J^rY$ when evaluated along a prolonged section in the sense that, given a form $\rho\in\Omega(J^rY)$,
\begin{equation}
\diff_M(J^r\gamma^* \rho) = J^{r+1}\gamma^* \diff_H\rho,
\end{equation}
where $\diff_M : \Omega(M) \to \Omega(M)$ is the exterior derivative on the base manifold $M$.

Furthermore, a Lagrangian $\lambda\in \Omega_{m,M}(V^r)$ will be variationally trivial (have an identically vanishing Euler-Lagrange form) if and only if there exist a form $\mu\in\Omega_{m-1}(V^{r-1})$ such that $\lambda = \diff_H \mu$, \cite{Krupka-book} p. 134. This requirement is further equivalent to the condition that, for any fibered chart $(V,\psi)$, there exist functions $g^\mu : V^r \to \mathbb{R}$ such that $\lambda = \left(\diff_\mu g^\mu\right) \diff x$, \cite{Krupka-book} p. 134.

\subsection{Canonical variational completion}

In the following, we will only study \textit{local} variationality; thus, we can assume with no loss of generality that $V^{r}\subset \mathbb{R}^{N}$ (for an appropriate value of $N$) and omit the explicit mention of $\psi^{r}.$

Consider now an arbitrary PDE system (\ref{PDE}) and build, as in \eqref{source_form}, the corresponding source form $\mathcal{E}=\mathcal{E}_{A}\diff y^{A}\wedge \diff x\in \Omega_{m+1}(V^{r})$. We will assume from the beginning that the domain $V^{r}$ is \textit{vertically star-shaped} with center $\left( x^{\mu },0,\cdots,0\right) $; that is, for every point $(x^{\mu},y^{A},y_{~\nu}^{A},\cdots,y_{~\nu _{1}\cdots\nu _{r}}^{A})\in V^{r},$ the whole segment $(x^{\mu },ty^{A},ty_{~\nu }^{A},\cdots,ty_{~\nu _{1}\cdots\nu _{r}}^{A})$, $t\in \lbrack 0,1]$, joining the center with the given point, lies in  $V^{r}$.

Under this assumption, one can introduce on the given chart domain $V^{r}$, the \textit{Vainberg-Tonti Lagrangian} $\lambda _{\mathcal{E}}=\mathcal{L}_{\mathcal{E}}\diff x,$ where\footnote{The idea behind the Vainberg-Tonti Lagrangian construction is the same as the one in the proof of the Poincar\'{e} Lemma, relating closed differential forms to exact ones : $\lambda _{\mathcal{E}}$ is actually the result of applying to $\mathcal{E}$ a \textit{homotopy operator} (the same as in the mentioned proof), see \cite{Krupka-book}, \cite{Olver}.}, \cite{Krupka-book} :
\begin{equation}
\mathcal{L}_{\mathcal{E}} \coloneqq y^{A}\int_{0}^{1}\mathcal{E}_{A}\circ \chi _{t}%
\diff t,  \label{eq:vtlag}
\end{equation}%
where $\chi _{t} :V^{r}\rightarrow V^{r}$ is given by
\begin{equation}
\chi _{t}(x^{\mu },y^{A},y_{~\nu }^{A},\cdots,y_{~\nu _{1}\cdots\nu_{r}}^{A})\coloneqq (x^{\mu },ty^{A},ty_{~\nu }^{A},\cdots,ty_{~\nu _{1}\cdots\nu_{r}}^{A}).  \label{def:chi_t}
\end{equation}

\vspace{50pt}
\textbf{Remarks:}
\begin{enumerate}
\item The Vainberg-Tonti Lagrangian is, in general, just defined on a coordinate chart. Yet, on \textit{tensor bundles} $Y$ -- which is our case of interest -- if the quantities $\mathcal{E}_{A}$ are tensor densities, then $\mathcal{L}_{\mathcal{E}}$ is a scalar density and $\lambda _{\mathcal{E}}=\mathcal{L}_{\mathcal{E}}\diff x$ is globally well defined. Actually, in this case, $\lambda _{\mathcal{E}}$ is generally covariant provided that $\mathcal{E}$ itself is so.

\item If the coordinate chart domain $V^{r}$ is not vertically star-shaped with center $\left( x^{\mu },0,\cdots,0\right) $ (which is obviously the case in gravity theories, where one cannot set all the field components $g_{\mu\nu }$ to zero), the Vainberg-Tonti Lagrangian will be understood as a limit :
\begin{equation}
\mathcal{L}_{\mathcal{E}}\coloneqq \underset{a\rightarrow 0}{\lim }\ y^A \int_{a}^{1}\mathcal{E}_{A}\circ \chi _{t}\diff t  \label{VT_Lag_limit}
\end{equation}%
and it makes sense whenever this limit is finite, see \cite{GB-var-completion}.
\end{enumerate}

\bigskip

The Vainberg-Tonti Lagrangian $\lambda _{\mathcal{E}}$ gives rise to the Euler-Lagrange expressions $\dfrac{\delta \mathcal{L}_{\mathcal{E}}}{\delta y^{A}}.$ These are related to the coefficients of $\mathcal{E}$ by
\begin{equation}
\frac{\delta \mathcal{L}_{\mathcal{E}}}{\delta y^{A}}=H_{A}+\mathcal{E}_{A},
\label{eq:helmholtz}
\end{equation}%
where $H_{A}$ are linear combinations of the coefficients $H_{AB}, H_{AB}^{\nu}, \cdots, H_{AB}^{\nu _{1}\cdots\nu _{r}}$ of the so-called \textit{Helmholtz form}, \cite{Krupka-book}, \cite{inv-problem-book}, whose vanishing is equivalent to the local variationality of $\mathcal{E}$ (see also Appendix \ref{appx:Helmholtz}, for their precise expressions in the case $r=2$). 
Consequently :
\begin{itemize}
\item If the original system is locally variational, then $\dfrac{\delta \mathcal{L}_{\mathcal{E}}}{\delta y^{A}}=\mathcal{E}_{A}$, \emph{i.e.}, the Vainberg-Tonti Lagrangian $\lambda_{\mathcal{E}}$ is a Lagrangian for the original system,

\item The quantities $H_{A}$ measure the ``obstructions from variationality'' of $\mathcal{E}_{A}$.
\end{itemize}

The following definition thus makes sense :

\begin{definition}
, \cite{canonical-var-completion}: Given a PDE system (\ref{PDE})
(equivalently, \eqref{eq:sourceformonchart}), its canonical
variational completion is the Euler-Lagrange system%
\begin{equation}
\dfrac{\delta \mathcal{L}_{\mathcal{E}}}{\delta y^{A}}\circ J^{r}\gamma =0
\end{equation}
of its attached Vainberg-Tonti Lagrangian.
\end{definition}

In other words, the canonical variational completion of a PDE system is obtained by adding, as correction terms, the ``obstruction from local variationality'', $H_{A}$, to the respective system; one Lagrangian for the ``corrected'' (now variational) system being the Vainberg-Tonti Lagrangian \eqref{eq:vtlag}.

\bigskip

\textbf{Example : Einstein tensor as the canonical variational completion of the Ricci tensor, \cite{canonical-var-completion}.}\newline
As already mentioned in the introduction, a motivation for the above construction was given by the historically first variant of gravitational field equations proposed by Einstein
\begin{equation}
R_{\mu \nu }=\kappa T_{\mu \nu }.  \label{eq:oldefe}
\end{equation}%
The configuration bundle is, in this case, the fibered manifold $(\met\left(M\right) ,\pi ,M),$ where $\met(M)$ is the set of symmetric and nondegenerate tensors of type $\left( 0,2\right) $ over a given manifold $M.$ Its sections are metric tensors $\metric$, locally described by $\left( x^{\rho }\right) \mapsto \left(g_{\mu \nu }\left( x^{\rho }\right) \right) .$ On the jet bundle $J^{2}\met(M)$, one can thus consider as local coordinates $\left( x^{\mu };g_{\mu \nu}, g_{\mu \nu ,\rho }, g_{\mu \nu ,\rho \sigma }\right) .$ With this choice, the left hand side of (\ref{eq:oldefe}) can be encoded (after raising the indices and densitizing) into the invariant source form
\begin{equation}
\mathcal{E}=R^{\mu \nu }\sqrt{\left\vert \det g\right\vert }\diff g_{\mu \nu
}\wedge \diff x\in \Omega _{m+1}\left( J^{2}\met\left(M\right) \right) ,
\end{equation}%
where $R^{\mu \nu }$ is the formal Ricci tensor (the word ``formal'' means that $R^{\mu \nu }$ is calculated by the usual formula but, in this expression, $g_{\mu \nu }$, $g_{\mu \nu ,\rho }$ and $g_{\mu \nu ,\rho \sigma }$ are regarded as independent coordinate functions on $J^{2}\met\left(M\right) ,$ not as functions of $x^{\mu}$).

A direct calculation then shows that the Vainberg-Tonti Lagrangian of $\mathcal{E}$ is actually the Einstein-Hilbert Lagrangian
\begin{equation}
\label{eq:LEH}
\lambda _{g} \coloneqq R \sqrt{\left\vert\det(g)\right\vert} \diff x,
\end{equation}
where $R$ is the formal Ricci scalar, leading to the Einstein equations (\ref{efe}).

\bigskip

Other applications of the canonical variational completion algorithm studied so far are, \emph{e.g.}, symmetrisation of canonical energy-momentum tensors in the special-relativistic limit \cite{canonical-var-completion}, Finsler gravity \cite{Finsler-var-completion}, and Gauss-Bonnet gravity \cite{GB-var-completion}.

\section{What if we only know a part of the equations?}\label{sec:bootstrap}

The above procedure is helpful as such in the case of purely metric theories of gravity. Yet, in theories of gravity employing a metric and another variable (\emph{e.g.}, a scalar field or a connection), one often has an ``educated guess'' on the form of the metric equations only -- \emph{e.g.}, resemblance with the Einstein equations with a cosmological constant. For such situations, we prove below that one can recover the Lagrangian, up to boundary terms and terms that do not involve the metric or its derivatives; as a byproduct, we find the variationally completed metric equations, which are the closest variational equations to our initial guess. Possible Lagrangian terms that are independent of the metric remain to be found by other means.

\bigskip

More generally, assume one wants to build a variational theory involving  \textit{two} groups of dynamical variables, say $y^{A}=y^{A}\left( x^{\mu}\right) $ and $z^{I}=z^{I}\left( x^{\mu }\right) $, but only have an educated guess about the field equations with respect to $y^{A}$ :
\begin{equation}
\mathcal{E}_{A}\Big(x^{\mu };y^{B}, \partial_{\nu} y^{B}, \cdots, \partial_{\nu _{1}}\cdots\partial_{\nu _{r}}y^{B}; z^{I}, \partial_{\nu}z^{I}, \cdots, \partial_{\nu_{1}}\cdots\partial_{\nu_{r}}z^{I}\Big) = 0,  \label{eq:PDE_syst}
\end{equation}%
where the number of equations is equal to the number of  $y^{A}$-variables. In the following, we will try to recover the variationally completed $y^{A}$-equations and, in so far as possible, the missing $z^{I}$-equations.

\subsection{Variational bootstrapping}

The configuration manifold corresponding to the situation described above is a fibered product manifold
\begin{equation}
Y=Y_{1}\times _{M}Y_{2},  \label{def:Y}
\end{equation}%
where $\times_M$ denotes the fibered product, which has as fibered charts $\left( V,\psi \right) ,$ $\psi =\left( x^{\mu};y^{A},z^{I}\right)$.
\\On the jet bundle $J^{r}Y\equiv J^{r}Y_{1}\times_{M}J^{r}Y_{2},$ we will denote the naturally induced charts $\left(V^{r},\psi ^{r}\right) ,$ with $\psi ^{r}=(x^{\mu };y^{A},y_{~\nu}^{A},\cdots,y_{~\nu _{1}\cdots\nu _{r}}^{A};z^{I},z_{~\nu }^{I},\cdots,z_{~\nu_{1}\cdots\nu _{r}}^{I})$ and by 
\begin{equation}
p_{1}:J^{r}Y\rightarrow J^{r}Y_{1},~\ \ \ p_{2}:J^{r}Y\rightarrow J^{r}Y_{2},
\end{equation}%
the projections onto the two factors of $J^{r}Y.$

Any source form on $V^{r}$ is then uniquely written as a sum of two components 
\begin{equation}
\mathcal{E} = \mathcal{E}_{A}\diff y^{A}\wedge \diff x+\mathcal{E}_{I}\diff z^{I}\wedge \diff x=:\mathcal{E}_{\left( 1\right) }+\mathcal{E}_{\left( 2\right) }\in \Omega _{m+1}\left(V^{r}\right) ,  \label{E_2_split}
\end{equation}%
where $\mathcal{E}_{\left( 1\right) }$ is $p_{2}$-horizontal (\emph{i.e.}, contains no terms in $\diff z^{I},\cdots,dz_{~\nu_{1}\cdots \nu_{r}}^{I}$) and $\mathcal{E}_{\left(2\right) }$ is $p_{1}$-horizontal; the functions $\mathcal{E}_{A}$ and $\mathcal{E}_{I}$ each depend, in principle, on all the coordinates $(x^{\mu };y^{B},y_{~\nu}^{B},\cdots,y_{~\nu _{1}\cdots\nu _{r}}^{B};z^{I},z_{~\nu }^{I},\cdots,z_{~\nu_{1}\cdots\nu _{r}}^{I}).$

Our postulated equations (\ref{eq:PDE_syst}) are thus completely encoded into the component $\mathcal{E}_{\left( 1\right) }$ of $\mathcal{E};$ we do not have any input information about $\mathcal{E}_{(2)}.$ Hence, the best thing we can do is to variationally complete $\mathcal{E}_{\left( 1\right) },$ as follows :

\bigskip

Fix a vertically star-shaped fibered chart $\left( V^{r},\psi ^{r}\right) $ on $J^{r}Y$ and define the \textit{partial fiber homothety} $\chi _{t,1}\coloneqq \left( \chi_{t},\id\right) :V^{r}\rightarrow V^{r},$ acting on the variables $y^{A},y_{~\nu }^{A},\cdots,y_{~\nu _{1}\cdots\nu _{r}}^{A}$ only
\begin{equation}
\chi _{t,1}(x^{\mu },y^{A},\cdots,y_{~\nu _{1}\cdots\nu_{r}}^{A},z^{I},\cdots,z_{~\nu _{1}\cdots\nu _{r}}^{I})\coloneqq (x^{\nu},ty^{A},\cdots,ty_{~\nu_{1}\cdots\nu_{r}}^{A},z^{I},\cdots,z_{~\nu_{1}\cdots\nu_{r}}^{I})
\end{equation}%
and, accordingly, the partial Vainberg-Tonti Lagrangian $\lambda _{1}= \mathcal{L}_{1}\diff x\in \Omega _{m,M}\left( J^{r}Y\right) $ as
\begin{equation}
\mathcal{L}_{1} \coloneqq y^{A}\int_0^1 \mathcal{E}_{A}\circ \chi _{t,1}\diff t.  \label{def_lambda1}
\end{equation}%
We then obtain the following result :

\begin{theorem}
\label{thm:partial_var_compl}(\textbf{Variational bootstrapping}): If
the partial differential system \eqref{eq:PDE_syst} is locally variational and the Vainberg-Tonti-type Lagrangian $\lambda _{1}=\mathcal{L}_{1}\diff x$ as in \eqref{def_lambda1} can be defined, then :

\begin{enumerate}
\item $\lambda _{1}$ is a locally defined Lagrangian for (\ref{eq:PDE_syst});

\item Any other Lagrangian $\lambda $ producing (\ref{eq:PDE_syst}) as $y^{A}$-field equations can only differ from $\lambda _{1}$ by a $y^{A}$-independent term and a boundary term :
\begin{equation}\label{eq:partial_var_compl}
\lambda =\lambda _{1}+\lambda _{2}+\lambda _{0},
\end{equation}%
where $\lambda _{2}=\mathcal{L}_{2}(x^{\mu };z^{I},z_{~\mu }^{I},\cdots,z_{~\mu_{1}\cdots\mu _{r}}^{I})\diff x$ only (\emph{i.e.}, $\lambda _{2}$ is projectable onto the second factor $J^{r}Y_{2}$) and $\lambda_{0}=\left( \diff_{\mu }f^{\mu }\right)
\diff x$ is a boundary term.
\end{enumerate}
\end{theorem}

\begin{proof}
As discussed above, the configuration system for (\ref{eq:PDE_syst}) is the fibered product $Y=Y_{1}\times _{M}Y_{2},$ as in (\ref{def:Y}). Fix a vertically star-shaped fibered chart domain $V^{r}\subset J^{r}Y$ and encode the left-hand sides of our equations into the source form $\mathcal{E}_{\left( 1\right) }=\mathcal{E}_{A}\diff y^{A}\wedge \diff x$ on $V^{r}$.

\begin{enumerate}
\item Proving that $\lambda _{1}$ is indeed a Lagrangian for $\mathcal{E}_{\left(1\right) }$ is done by direct computation, in a completely similar way to, \emph{e.g.}, \cite{Krupka-book}, Sec. 4.9, or \cite{GB-var-completion} -- with the only difference that here we have an extra variable $z^{I},$ which, yet, remains unaffected. For courtesy to the reader, we briefly reproduce in Appendix \ref{appx:Helmholtz} the calculation for $r=2.$ 

\item Since $\lambda $ and $\lambda _{1}$ should produce the same Euler-Lagrange expressions with respect to $y^A$, it follows that they can only differ by terms that do not contribute in any way to these expressions -- that is, terms that do not have any dependence on $y^{A}, y^{A}_\mu, \cdots, y^{A}_{\mu_{1} \cdots \mu_{r}} $, or divergence expressions.
\end{enumerate}
\end{proof}

If the given equations (\ref{eq:PDE_syst}) are not locally variational, then, the partial variational completion procedure guarantees that the correction terms $\dfrac{\delta\mathcal{L}_{1}}{\delta y^{A}}-\mathcal{E}_{A}$ still have the meaning of obstructions from $y^{A}$-variationality of $\mathcal{E}_{\left( 1\right) }$.



\subsection{A first example: Einstein-Klein-Gordon equations}

Before passing to our case of interest, which are metric-affine theories of gravity, let us check how the variational bootstrapping method works in a case where a Lagrangian is already known. A straightforward such example is the one of a single real-valued scalar field $\phi$ minimally coupled to the metric $\metric$ in the context of General Relativity.

In this case, the dynamical variables are a metric tensor and a scalar field; the corresponding configuration space is thus
\begin{equation}
Y=\met\left(M\right) \times_{M}\left(M\times\mathbb{R}\right).
\end{equation}%
Accordingly, we can consider on $J^{2}Y$ the naturally induced fibered coordinates $\left( x^{\mu}; g_{\mu \nu },g_{\mu \nu,\rho }, g_{\mu \nu ,\rho \tau }; \phi ,\phi _{\rho },\phi _{\rho \tau }\right) $.
This situation is notoriously described by the Einstein-Klein-Gordon Lagrangian $\lambda _{\text{\tiny EKG}}=\mathcal{L}_{\text{\tiny EKG}}\diff x$, with
\begin{equation}
\mathcal{L}_{\text{\tiny EKG}}=\left( \dfrac{1}{2\kappa} R - \dfrac{1}{2} g^{\alpha \beta} \phi _{\alpha}\phi _{\beta} - V\left(\phi\right)\right) \sqrt{\left\vert \det g\right\vert },  \label{L:EKG}
\end{equation}%
(where $\kappa \in \mathbb{R}_0$ is a constant and $V=V\left( \phi \right)$ is a real-valued smooth function), which produces the field equations
\begin{equation}
\left\{ 
\begin{array}{c}
G_{\mu \nu }=\kappa T_{\mu \nu }^{\left( \phi \right) } \\ 
\square \phi =V^{\prime }\left( \phi \right) 
\end{array}%
\right. ,  \label{eq:EKG}
\end{equation}%
where $\square = g^{\mu \nu }\nabla _{\mu }\nabla _{\nu }$ is the covariant d'Alembertian and%
\begin{equation}
T_{\mu \nu }^{\left(\phi\right)}=\phi _{\mu}\phi _{\nu} - \dfrac{1}{2} g^{\alpha\beta}\phi_{\alpha}\phi_{\beta} g_{\mu\nu} - V\left(\phi\right) g_{\mu\nu}.  \label{eq:Tphi}
\end{equation}

Let us pretend for one moment to have no idea about the Lagrangian (\ref{L:EKG}) and recover it  from the metric equations (\ref{eq:EKG}) by $\met\left(M\right)$-variational completion. 
Considering $g_{\mu \nu }$ as our dynamical variables, the relevant source form is $\mathcal{E}_{g}=\mathcal{E}^{\mu \nu }\diff g_{\mu \nu }\wedge \diff x$, where
\begin{equation}
\mathcal{E}^{\mu \nu } \coloneqq -\dfrac{1}{2}\left(\dfrac{1}{\kappa} G^{\mu\nu} - \phi^{\mu}\phi^{\nu} + \dfrac{1}{2} g^{\alpha\beta} \phi_{\alpha}\phi_{\beta}g^{\mu \nu } + V\left(\phi\right) g^{\mu\nu}\right) \sqrt{\left\vert \det g\right\vert },
\end{equation}%
with $\phi^\alpha \coloneqq g^{\alpha\beta}\phi_\beta$. The Vainberg-Tonti Lagrangian $\lambda =\mathcal{L}_{\phi }\diff x$ is then given by : %
\begin{equation}
\mathcal{L}_{\phi }=g_{\mu\nu}\int_0^1 \mathcal{E}^{\mu \nu }\circ \chi _{t,1}\diff t,  \label{VT_Lagrangian_phi}
\end{equation}%
where the lower integration point is understood as a limit and the partial fibered homotheties \[\chi _{t,1}:\left( x^{\mu},g_{\mu \nu },g_{\mu \nu ,\rho }, g_{\mu \nu ,\rho \tau}; \phi, \phi _{\rho }, \phi_{\rho \tau }\right) \mapsto \left( x^{\mu},t g_{\mu \nu }, t g_{\mu \nu ,\rho }, t g_{\mu \nu ,\rho \tau }; \phi, \phi _{\rho }, \phi _{\rho \tau}\right)\] only affect $g_{\mu \nu }$ and their derivatives. We then have
\begin{eqnarray}
g^{\mu \nu }\circ \chi _{t,1} \hspace{-8pt}&=&\hspace{-8pt} t^{-1}g^{\mu\nu},\sqrt{\left\vert\det g\right\vert }\circ \chi_{t,1} = t^{m/2}\sqrt{\left\vert\det g\right\vert },~\\ \lc{\Gamma}_{\nu\rho}^{\mu}\circ\chi_{t,1} \hspace{-8pt}&=&\hspace{-8pt} \lc{\Gamma}_{\nu\rho}^{\mu},~\ \ud{\lc{R}}{\mu}{\nu \rho \sigma}\circ \chi _{t,1} = \ud{\lc{R}}{\mu}{\nu \rho \sigma},~\ \lc{R}_{\nu \sigma }\circ \chi _{t,1} = \lc{R}_{\nu \sigma },~\ \ \lc{R}^{\nu\sigma }\circ \chi _{t,1} = t^{-2} \lc{R}^{\nu \sigma },
\end{eqnarray}
where $\lc{\Gamma}_{\nu\rho}^{\mu}$, $\ud{\lc{R}}{\mu}{\nu \rho \sigma}$ and $\lc{R}_{\mu \nu }$ are, respectively, the formal Christoffel symbols, formal curvature tensor and formal Ricci tensor of the Levi-Civita connection associated to $\metric$.
This then gives us~:
\begin{eqnarray}
\mathcal{L}_{\phi } \hspace{-4pt}&=&\hspace{-4pt}\displaystyle -\dfrac{1}{2} g_{\mu\nu} \int_0^1 \left\{\left(\dfrac{1}{\kappa}G^{\mu \nu} - \phi^{\mu}\phi^{\nu} + \dfrac{1}{2}g^{\alpha\beta}\phi_{\alpha}\phi_{\beta}g^{\mu \nu}\right) t^{\frac{m}{2}-2} + V\left(\phi\right) g^{\mu\nu} t^{\frac{m}{2}-1}\right\} \sqrt{\left\vert \det g\right\vert }\diff t \\ \hspace{-4pt}&=&\hspace{-4pt} \left(\dfrac{1}{2\kappa }R - \dfrac{1}{2}g^{\alpha\beta}\phi_{\alpha}\phi_{\beta} - V\left(\phi\right) \right) \sqrt{\left\vert \det g\right\vert} = \mathcal{L}_{\text{\tiny EKG}},
\end{eqnarray}%
as expected. Variation with respect to $\phi$ then allows to also recover the second equation of \eqref{eq:EKG}, the Klein-Gordon equation.

\bigskip 
According to Theorem \ref{thm:partial_var_compl}, Lagrangians that cannot be recovered by $\met(M)$-variational completion are either variationally trivial (boundary terms), or metric-independent. Let us investigate the second possibility, corresponding to the $\lambda_2$-term of \eqref{eq:partial_var_compl}. In other words, we are looking for Lagrangians that could induce modifications of the second equation of \eqref{eq:EKG} without altering the first one. Such Lagrangians $\lambda_2=\mathcal{L}_{2}\diff x$ should be horizontal $m$-forms built using $\phi$ and its partial derivatives only. The only \textit{natural (generally covariant)} operators we can use, without employing the metric or its derivatives, are the exterior derivative and the wedge product (see also the discussion in the beginning of Section \ref{sec:classif_affine} below). From the coordinate function $\phi$, one can build $\diff_H\phi  = \phi_\mu  \diff x^\mu\in \Omega_{1,M}(J^1Y)$ but this does not allow to go further, considering that $\diff_H\left(\diff_H\phi\right)\equiv 0$ and $\diff_H\phi\wedge\diff_H\phi\equiv 0$. \\
In other words, if $\dim(M)>1$ (which is our case of interest), there are no natural Lagrangians that can be built from a scalar field and its partial derivatives alone. Of course, one could think, \emph{e.g.}, of building an $m$-form by multiplication of (a function of) $\phi$ by an invariant volume form $\tensi{\epsilon}$; yet, there is no \textit{natural} choice for such an $m$-form -- at least, not without resorting to extra geometric structures such as a metric or a connection. \\
To conclude, in the case of a single, real-valued scalar field $\phi$ minimally coupled to the metric $\metric$, any natural Lagrangian can be fully recovered, up to boundary terms, from the metric equations only.\footnote{Actually, the metric equations do not need to be the Einstein equations. The argument works for any natural theory employing a metric and a scalar field.}

\section{The case of metric-affine theories}\label{sec:metric-affine}

In metric-affine theories, the dynamical variables are, \textit{a priori,} a metric $\metric$ and an independent connection $\Gamma .$ Yet, since, on the one hand, a metric automatically determines its Levi-Civita connection -- which we will in the following denote by $\lc{\Gamma}$ -- and, on the other hand, the difference between two connections is tensorial, it is customary to split $\Gamma $ as
\begin{equation}
\Gamma =~\lc{\Gamma} + \tens{L},
\end{equation}%
where $\mathbf{L}$ is a tensor field of type $\left( 1,2\right) $ over the spacetime manifold $M$ called the distortion tensor. This way, the problem of determining the pair $\left( \metric,\Gamma\right) $ is equivalent to the one of determining the pair $\left(\metric,\mathbf{L}\right).$ The latter pairs are local sections of the fibered product :
\begin{equation}
Y=\met\left(M\right) \times_{M}T_{2}^{1}M ,
\end{equation}%
where $T_{2}^{1}M$ is the vector bundle of all tensors of type $\left(1,2\right) $ over $M$; with local fibered coordinates, say, $(x^{\mu };g_{\mu \nu },\ud{L}{\mu}{\nu \rho})$. A general source form of order $r$ on $Y$ will then have a $g_{\mu \nu}$-part and an $\ud{L}{\mu}{\nu\rho}$ one :
\begin{equation}
\mathcal{E} = \mathcal{E}^{\mu \nu }\diff g_{\mu \nu}\wedge \diff x+\du{\mathcal{E}}{\mu}{\nu\rho}\diff \ud{L}{\mu}{\nu\rho}\wedge \diff x=:\mathcal{E}_{g}+\mathcal{E}_{L}.
\end{equation}

In the following, we assume that we have an educated guess for the metric equations $\mathcal{E}_{\mu\nu}=0,$ that is, we know
\begin{equation}
\mathcal{E}_{g}=\mathcal{E}^{\mu \nu }\diff g_{\mu \nu}\wedge \diff x.
\end{equation}%

\bigskip
Assuming that we are only looking for generally covariant metric-affine Lagrangians, the following result is known :

\begin{theorem}
\label{thm:Janyska}(Janyska, \cite{Janyska}) : All generally covariant Lagrangians $\lambda =\mathcal{L}\diff x$, of order $r$ in a metric tensor and another tensor variable $\tensi{\Phi}$ are given by :
\begin{equation}
\mathcal{L}=\mathcal{L}\left(\lc{\nabla}^{(r-2)}\tens{R}(\lc{\Gamma}),\lc{\nabla}{}^{(r)}\tensi{\Phi}\right),
\label{eq:Janyska}
\end{equation}%
where $\lc{\nabla}{}^{\left(r\right)} = \left(\id, \lc{\nabla}, \lc{\nabla}\lc{\nabla}, \cdots, \lc{\nabla}{}^{r}\right)$ denotes covariant $\lc{\Gamma}$-derivatives up to order $r$ and $\tens{R}(\lc{\Gamma})$ is the (formal) curvature tensor of $\lc{\Gamma}.$
\end{theorem}

Picking $\tensi{\Phi} = \mathbf{L}$ leads to an immediate consequence :

\begin{corollary}
\label{cor:metric_affine_L} All the variationally nontrivial (non-boundary) terms of a generally covariant metric-affine Lagrangian containing the metric tensor or its derivatives can be recovered by $\met(M)$-variational completion. The only terms that cannot be recovered by this procedure are purely affine terms; these are built from the distortion tensor and its Levi-Civita covariant derivatives, in such a way that the Christoffel symbols (and their derivatives) eventually cancel out. 
\end{corollary}

Let us explore, in the following, all the possibilities of building natural purely affine Lagrangians, over 4-dimensional spacetimes.

\section{Purely affine invariants}\label{sec:torsion}

Fix $\dim M=4$. We will determine all generally covariant (or \textit{natural}) Lagrangians
\begin{equation}
\lambda =\mathcal{L}(\ud{L}{\alpha}{\beta\gamma},\ud{L}{\alpha}{\beta\gamma,\mu}, \cdots, \ud{L}{\alpha }{\beta\gamma,\mu _{1}\cdots\mu _{r}})\diff x
\label{lambda_L_partial_derivs}
\end{equation}%
which can be built on 4-dimensional metric-affine backgrounds $\left(M,\metric, \Gamma\right) $, from the components of the distortion tensor $\mathbf{L}$ of the connection $\Gamma = \lc{\Gamma}+\mathbf{L}$ and their derivatives -- briefly, on $L,\partial L,\partial \partial L,\cdots$ alone -- \emph{i.e.} that are completely independent of the metric $\metric$. In other words, we are looking for a differential form 
\begin{equation}
\lambda \in \Omega _{4,M} (J^{r}T_{2}^{1}M).
\end{equation}%
The technique we use below is the so-called algebraic method for finding differential invariants by Kolar, Michor and Slovak, \cite{Kolar}.

\subsection{The polynomial property}

We will first prove that natural Lagrangians depending on distortion alone must actually be polynomial in the components of $\mathbf{L}$ and their derivatives. To this aim, we start by using Theorem \ref{thm:Janyska} in the previous section, which shows that any natural metric-affine Lagrangian $\lambda$ must, apart from derivatives of the curvature tensor $\tens{R}(\lc{\Gamma})$, be expressed as a smooth function of $\ud{L}{\alpha}{\beta\gamma}$ and their $\lc{\Gamma}$\textit{-(formal) covariant derivatives}\footnote{where, by formal covariant derivatives, we mean expressions given by the same formulas as usual covariant derivatives, just, applied to the independent coordinate functions $\ud{L}{\alpha}{\beta\gamma}, \ud{L}{\alpha}{\beta\gamma, \delta}, \cdots,$ on $J^{r}Y$.} up to some order $r$. According to Corollary \ref{cor:metric_affine_L}, purely affine terms must then be of the form  : 
\begin{equation}
\lambda =\mathcal{\hat{L}}(\ud{L}{\alpha}{\beta\gamma}, \lc{\nabla}_{\mu} \ud{L}{\alpha}{\beta\gamma}, \cdots, \lc{\nabla}_{\mu _{1}}\cdots\lc{\nabla}_{\mu _{_{r}}} \ud{L}{\alpha}{\beta\gamma})\diff x.  \label{lambda_L_tensor}
\end{equation}%

The advantage of the latter writing is that all the building blocks of $\lambda $ are tensor fields; in other words, we transfer our problem of finding a mapping $\lambda $ defined on the fibers of a jet bundle, into one of finding a function defined on a Cartesian product of \textit{vector spaces}.
More precisely, the formal covariant differentiation operator $\lc{\nabla} : T_{2}^{1}M \rightarrow T_{2}^{1}M \otimes T^{\ast }M$ acts linearly on the fibers of $T_{2}^{1}M $ (to be even more precise, $\lc{\nabla}$ is a vector bundle morphism covering the identity of $M$) and thus its image is again a vector bundle; similarly, the images of further iterations of $\lc{\nabla}$ are again vector bundles. That is, fixing an arbitrary point $x\in M,$ and a chart around it, we obtain that the fibers
\begin{equation}
V_{0}\coloneqq \left( T_{2}^{1}M \right) _{x},~\ V_{1}\coloneqq \left( \lc{\nabla}T_{2}^{1}M \right) _{x},\cdots,V_{r}\coloneqq \left(\lc{\nabla}{}^{r}T_{2}^{1}M \right)_{x},\ 
\end{equation}%
where $\lc{\nabla}{}^{r} = \lc{\nabla}\cdots\lc{\nabla}$ with $r$ terms $\lc{\nabla}$, are all finite dimensional real vector spaces, whereas the restriction of the Lagrangian density to these fibers becomes a mapping
\begin{equation}\label{eq:Latx}
f\coloneqq \mathcal{\hat{L}}_{x}:V_{0}\times V_{1}\times \cdots\times V_{r}\rightarrow \mathbb{R}
\end{equation}%
defined on a Cartesian product of vector spaces. This allows us to use the following result :

\begin{theorem}
\label{thm:homog}(\textbf{Homogeneous Function Theorem}, \cite{Kolar}). Let $V_{i},$ $i=0,\cdots,r$ be finite dimensional real vector spaces. If $f :V_{0}\times \cdots\times V_{r}\rightarrow \mathbb{R}$ is a smooth function with the property that there exist $b\in \mathbb{R}$ and $a_{i}>0,$ $i=0, \cdots, r$ such that: 
\begin{equation}
k^{b}f\left( v_{0},\cdots,v_{r}\right) =f\left(k^{a_{0}}v_{0},\cdots,k^{a_{r}}v_{r}\right) ,~\ \ \ \ \forall k>0,
\end{equation}%
then, $f$ must be a sum of polynomials of degree $d_{i}$ in $v_{i},$ where $d_{i}\in \mathbb{N}$ satisfy the relation
\begin{equation}
a_{0}d_{0} + \cdots + a_{r}d_{r}=b.
\end{equation}%
If there are no non-negative integers $d_{0},\cdots,d_{r}$ with the above property, then $f$ is the zero function.
\end{theorem}

\bigskip
Let us apply the homogeneous function theorem to our situation. The requirement that $\lambda =\mathcal{\hat{L}}\diff x$ should be invariant under any coordinate changes naturally induced by local coordinate changes $x^{\mu}=x^{\mu }\big(x^{\nu ^{\prime}}\big)$ on the base manifold implies, in particular, the invariance under homotheties (with constant factor $k>0$)
\begin{equation}
x^{\alpha }=kx^{\alpha ^{\prime}},~\ \ \alpha =0,1,2,3.  \label{homothety}
\end{equation}%
Under these transformations, the wedge product $\diff x = \diff x^{0}\wedge \diff x^{1}\wedge\diff x^{2}\wedge \diff x^{3}$ changes as
\begin{equation}
\diff x=k^{4}\diff x^{\prime}.
\end{equation}%
The invariance condition on $\lambda =\mathcal{\hat{L}}\diff x$ then implies
\begin{equation}
\mathcal{\hat{L}}(L^{\prime}, \lc{\nabla}L^{\prime}, \cdots, \lc{\nabla}{}^{r}L^{\prime})=k^{4}\mathcal{\hat{L}(}L,\lc{\nabla}L, \cdots, \lc{\nabla}{}^{r}L),
\end{equation}%
where primes on $L$ denote the components of $L$ and of its covariant derivatives in the new coordinates $x^{\alpha ^{\prime}}$ (and we have omitted the indices for simplicity). Moreover, we have that
\begin{equation*}
L^{\prime} = kL, \lc{\nabla}L^{\prime} = k^{2}L, \cdots, \lc{\nabla}^{r}L^{\prime} = k^{r+1}\lc{\nabla}{}^{r}L,
\end{equation*}%
which then leads to:
\begin{equation}
k^{4}\mathcal{\hat{L}}(L,\lc{\nabla}L, \cdots, \lc{\nabla}{}^{r}L)=\mathcal{\hat{L}}(kL,k^{2}\lc{\nabla}L, \cdots, k^{r+1}\lc{\nabla}{}^{r}L).
\end{equation}

Now, fix an arbitrary point $x\in M$. The restriction of $\mathcal{\hat{L}}$ to the fiber at $x$ of our configuration space is a smooth mapping defined on a Cartesian product of vector spaces \eqref{eq:Latx}. Applying to it the homogeneous function theorem, with
\begin{equation}
b = 4,~\ a_{0} = 1,~\ a_{1} = 2, \cdots, a_{r} = r+1,
\end{equation}%
it follows that $\mathcal{\hat{L}}$ must be a sum of $r$ homogeneous polynomials in $L,\lc{\nabla}L$ etc., whose degrees $d_{i}\in \mathbb{N}$ in the derivatives of order $i$ in $L$ satisfy
\begin{equation}
d_{0} + 2 d_{1} + 3 d_{2} + \cdots + \left(r+1\right) d_{r}=4.
\end{equation}%
In particular, we must have $r\leq 3$. We have then proven :

\begin{proposition}
Any smooth, generally covariant Lagrangian $\lambda =\mathcal{\hat{L}}\diff x$ depending  only on the distortion tensor of an affine connection $\Gamma$ and its derivatives is at order at most three. Moreover, the Lagrangian density $\mathcal{\hat{L}}$ must be expressed as a sum of homogeneous polynomials of degrees $d_{0}$, $d_{1}$, $d_{2}$ and, respectively, $d_{3}$ in the variables $L,\lc{\nabla}L$, $\lc{\nabla}\lc{\nabla}L$ and, respectively, $\lc{\nabla}\lc{\nabla}\lc{\nabla}L$ satisfying :
\begin{equation}
\label{eq:degree}
d_{0} + 2 d_{1} + 3 d_{2} + 4 d_{3}=4.
\end{equation}
\end{proposition}

\bigskip

The next step is to realize that, under the hypothesis that $\lambda $ cannot depend on the metric $\metric$, it cannot depend on either the coefficients $\lc{\Gamma}\overset{}{_{\beta\gamma}^{\alpha }}$ or their derivatives;
that is, $\lc{\Gamma}\overset{}{_{\beta \gamma}^{\alpha}}$ and their derivatives must eventually cancel out in the expression of $\mathcal{\hat{L}}$, giving
\begin{equation}
\mathcal{\hat{L}}(\ud{L}{\alpha}{\beta\gamma},\lc{\nabla}_{\mu}\ud{L}{\alpha}{\beta\gamma}, \cdots, \lc{\nabla}_{\mu _{1}} \cdots \lc{\nabla}_{\mu _{_{r}}}\ud{L}{\alpha}{\beta\gamma}) = \mathcal{L}(\ud{L}{\alpha}{\beta\gamma}, \ud{L}{\alpha}{\beta\gamma, \mu}, \cdots, \ud{L}{\alpha}{\beta\gamma, \mu_{1} \cdots \mu _{r}}).
\end{equation}
This leaves us with $\mathcal{L}$ as a sum of homogeneous polynomials in $L$, $\partial L$, $\partial \partial L$ and $\partial \partial \partial L$, of the same degrees as in \eqref{eq:degree}. We thus obtain :
\begin{corollary}
\label{cor:polynomial}Any smooth, generally covariant Lagrangian $\lambda =\mathcal{\hat{L}}\diff x$ depending  only on the distortion tensor of an affine connection $\Gamma$ and its derivatives is at order at most three. Moreover, it must be a sum of homogeneous polynomials, of degrees $d_{0}$, $d_{1}$, $d_2$ and, respectively, $d_3$ in the variables $L$, $\partial L$, $\partial \partial L$ and, respectively, $\partial \partial \partial L$ satisfying
\begin{equation}
\label{eq:degrees}
d_{0} + 2 d_{1} + 3 d_{2} + 4 d_{3}=4.
\end{equation}
\end{corollary}

\textbf{Remark.} In the Homogeneous Function Theorem, the smoothness assumption on $f$ on an entire Cartesian product of vector spaces $V_{0}\times \cdots\times V_{n}$ (in particular, smoothness at its zero vector) is essential. As a consequence, the result \textit{cannot} be applied to find Lagrangians depending on a metric tensor, since the condition $\det g \not= 0$ forbids $\left( g_{\alpha \beta }\right) $ from being all zero (otherwise stated, we cannot pick any of the fibers of $\met\left(M\right) $ as $V_{0},$ since these fibers are not vector spaces). This allows for non-polynomial natural Lagrangian forms in the metric, such as $\sqrt{\left\vert \det g\right\vert }\diff x.$

Yet, in our case, our configuration space $T_{2}^{1}M $ is a vector bundle, therefore the Homogeneous Function Theorem can be safely applied, ensuring that $\lambda $ must be polynomial, as in \eqref{eq:degrees}. Actually, as we will see in the next subsection, relation \eqref{eq:degrees} can be further simplified.

\subsection{Classification of pure distortion Lagrangians}\label{sec:classif_affine}

To find all generally covariant (natural), pure distortion Lagrangians $\lambda \in \Omega_{4,M}\left( J^{r}T_{2}^{1}M \right) ,$ we will use Corollary \ref{cor:polynomial} above, together with several known results in the literature, which we briefly review below :

\begin{enumerate}
\item Any natural Lagrangian (which is an equivariant mapping, under the action of the differential group), must be obtained as the result of a  \textit{natural operator}, \cite{Janyska}.

\item\label{rmk:dh} The only possible first order natural operators acting on a differential $p$-form on a manifold and returning a $(p+1)$-form, are constant multiples of the exterior derivative $\diff$, \cite{Janyska}, \cite{Kolar}. Moreover, natural (generally covariant) differential operators on arbitrary tensor bundles are (see, \emph{e.g.}, \cite{Slovak-thesis}, p4) compositions of exterior differentiation and invariant (natural) algebraic operators. Actually, when producing Lagrangians, which are \textit{horizontal} differential forms, the appropriate operator is the total (or horizontal) exterior derivative $\diff_H = h \circ \diff$. In particular, since $\diff_H\circ\diff_H \equiv 0$, there are no natural operators of order $r > 1$ on tensor bundles, involving the tensor variables and their derivatives alone.

\item Natural algebraic operators on tensor bundles are,  \cite{Slovak-thesis}, only finite iterations of : permutations of indices, tensor product with invariant tensors, trace with respect to one subscript and one superscript and linear combinations of these.
\end{enumerate}

\bigskip

Using the above mentioned results, together with Corollary \ref{cor:polynomial}, we find
\begin{theorem} \label{lem:degrees_lambda}
Assume $\dim(M)=4$. Then, all natural metric-affine Lagrangians depending on the distortion of the connection alone are of order at most one and must be expressed as a sum whose terms fall into one of the following classes:

\begin{enumerate}
\item Purely algebraic terms: These must be expressed as homogeneous polynomials of degree 4 in the components $\ud{L}{\alpha}{\beta \gamma}$ of the distortion tensor.

\item First order terms: These must be either quadratic in $L$ and linear in $\partial L$, or quadratic in $\partial L$ (and independent of $L$). In any of the two cases, these must be obtained via horizontal exterior differentiation and/or wedge product, from differential forms of rank at most three, depending algebraically on $L$.
\end{enumerate}
\end{theorem}

\begin{proof}
Assume $\lambda =\mathcal{L}\diff x$ is a natural Lagrangian of order $r$ on $T_{2}^{1}M .$ According to Corollary \ref{cor:polynomial}, we must have $r\leq 3$ and the density $\mathcal{L}=\mathcal{L}\left( L,\partial L,\partial \partial L,\partial \partial \partial L\right) $ must be polynomial in all its variables, with the respective degrees $d_{0}$, $d_{1}$, $d_{2}$, $d_{3}$ satisfying \eqref{eq:degrees}. The latter relation can obviously hold only for $d_{2},d_{3}\leq 1$ (where only one of the values $d_{2}$, $d_{3}$ can be non-zero) -- \emph{i.e.}, $\mathcal{L}$ is at most linear in one of the variables $\partial \partial L$ or $\partial \partial \partial L$; but, Result \ref{rmk:dh} above actually forbids this situation, meaning that we must actually have
\begin{equation}
d_{2} = d_{3} = 0.
\end{equation}%
In other words, $\lambda $ is of order at most one. Moreover, equation \eqref{eq:degrees} becomes :
\begin{equation}
d_{0} + 2 d_{1} = 4.
\end{equation}%
This leaves room for three possibilities :
\begin{enumerate}
\item $d_{0}=4, d_{1}=0$,
\item $d_{0} = 2, d_{1} = 1$,
\item $d_{0} = 0, d_{1} = 2$,
\end{enumerate}
corresponding to the two situations in the statement of the theorem.

The statement on the exterior derivative/wedge product structure for the latter two cases follows again from the results recalled in the beginning of this subsection.
\end{proof}

\bigskip

A complete list of the possibilities of building such Lagrangians is given in Appendix \ref{appx:classif}.

\section{Conclusion} \label{sec:Conclusion}

In the present paper, we have shown that, in any physical theory involving more than one dynamical variable, having an educated guess (which can be based, \emph{e.g.}, on some physical principle)  on the field equations with respect to one such variable, say $y^{A}$, one can find, in a systematic way, the closest variational equations $\dfrac{\delta \mathcal{L}}{\delta y^{A}} = 0$ to our original guess. Using this procedure, all the terms in the Lagrangian density $\mathcal{L}$ that involve, in a way or another, $y^{A}$ or their derivatives can be recovered. The obtained Lagrangian density then also determines the Euler-Lagrange equations for the other dynamical variables, say $z^I$, up to terms that are completely independent from $y^A$ and their derivatives (and hence do not contribute to the $y^A$ equation).

As a first application, we studied the case of metric-affine theories. Assuming that one has a guess on the \textit{metric} equations, one can recover the closest equations to these that are Euler-Lagrange equations for some Lagrangian. In that case, we also classified the terms in the Lagrangian that cannot be recovered by our technique (\emph{i.e.} terms that are completely independent of the metric tensor). We found that these are given by polynomial expressions (of degree at most 4) in the distortion of the connection and its first order derivatives.

\bigskip

In a forthcoming paper, we will apply this technique to determine the metric-affine models of gravity that give the closest metric equations to the $\Lambda$CDM model of cosmology.  This will then further allow us to constrain the evolution equations for the connection (\emph{i.e.}, for the distortion tensor).

\bigskip

Of course, the procedure outlined in this paper is fully general and can thus be applied to a variety of other contexts, \emph{e.g.}, to scalar-tensor theories, or scalar-vector-tensor theories.
To further extend our procedure it might also be interesting to study in detail cases where the convergence of the Vainberg-Tonti Lagrangian, as defined in this paper, cannot be guaranteed. We leave these and other questions for future works.

\bigskip

\textbf{Acknowledgments. }The authors are grateful to M. Hohmann for insightful discussions and for first raising the question of whether variational completion still makes sense when knowing just a part of the equations, and to C. Pfeifer, K. Hajian and colleagues from the London-Oldenburg online seminar for useful comments on metric-independent Lagrangians. L.D. was supported by a grant from the Transilvania Fellowship Program for Postdoctoral Research/Young Researchers (September 2022) and from a Postdoctoral sojourn grant from \textsc{Complexys} institute from University of Mons (Belgium). This work is based upon colaboration within the COST Actions : CA18108, ``Quantum Gravity Phenomenology in the Multi-Messenger Approach'' (QGMM) and CA 21136, ``Addressing observational tensions in cosmology with systematics and fundamental physics'' (CosmoVerse).

\appendix
\section{Helmholtz expressions and their relation to Vainberg-Tonti Lagrangian} \label{appx:Helmholtz}

In this appendix, we present the explicit form of the Helmholtz expressions, encoding obstruction to variationality, and detail the computations underlying the first statment of Theorem \ref{thm:partial_var_compl}. For simplicity, we here display everything in the case of a second order source form. Reasonings unfold similarly in the general case.

\subsection{Helmholtz expressions}

For a second order source form $\mathcal{E} = \mathcal{E}_{A} \diff y^{A} \wedge \diff x$, the local variationality conditions (Helmholtz conditions) read:
\begin{eqnarray}
H_{AB}^{~\mu \nu }:=\dfrac{\partial \mathcal{E}_{A}}{\partial y_{~\mu \nu}^{B}}-\dfrac{\partial \mathcal{E}_{B}}{\partial y_{~\mu \nu}^{A}}=0, \label{H1} \\
H_{AB}^{\mu}:=\dfrac{\partial \mathcal{E}_{A}}{\partial y_{~\mu }^{B}}+ \dfrac{\partial \mathcal{E}_{B}}{\partial y_{~\mu }^{A}}-2d_{\nu }\dfrac{\partial \mathcal{E}_{B}}{\partial y_{~\mu \nu }^{A}}=0,  \label{H2} \\
H_{AB}:=\dfrac{\partial \mathcal{E}_{A}}{\partial y^{B}}-\dfrac{\partial 
	\mathcal{E}_{B}}{\partial y^{A}}+d_{\mu }\dfrac{\partial \mathcal{E}_{B}}{%
	\partial y_{~\mu }^{A}}-d_{\mu }d_{\nu }\dfrac{\partial \mathcal{E}_{B}}{%
	\partial y_{~\mu \nu }^{A}} =0.  \label{H3}
\end{eqnarray}

\subsection{Complement to the proof of Theorem \ref{thm:partial_var_compl}}

Let us now discuss the computations underlying the first statment of Theorem \ref{thm:partial_var_compl}. To that purpose, we assume the same setup as in the statment of the theorem; yet, here -- for conciseness of the expressions -- we restrict our presentation to the case of a second order source form :

Assume that $\mathcal{E}_{(1)}  = \mathcal{E}_{A} \diff y^{A} \wedge \diff x$ is locally variational with respect to the variables $y^A$, that is, the corresponding Hemholtz expressions vanish. Denoting by $\lambda_{1} = \mathcal{L}_1 \diff x$ the Vainberg-Tonti Lagrangian of $\mathcal{E}_{(1)}$ with respect to the variables $y^A$ \eqref{def_lambda1}, we get (for $r=2$; the proof for general $r$ is completely similar): 
\begin{equation*}
\dfrac{\partial \mathcal{L}_{1}}{\partial y^{B}}=\int_0^1 \mathcal{E}_{B}\circ \chi _{t,1}\diff t+y^{A}\int_0^1 t\dfrac{\partial\mathcal{E}_{A}}{\partial y^{B}}\circ \chi _{t,1}\diff t,
\end{equation*}%
which, after integration by parts in the first term, leads to:
\begin{align*}
\dfrac{\partial \mathcal{L}_{1}}{\partial y^{B}}& =t\mathcal{E}_{B}(x^{\mu},ty^{A},ty_{~\mu }^{A},ty_{~\mu \nu }^{A},z^{I},z_{\ \mu }^{I},z_{~\mu \nu}^{I})\mid _{0^{{}}}^{1_{{}}} \\ & + y^{A} \int_0^1 t(\dfrac{\partial \mathcal{E}_{A}}{\partial y^{B}}-\dfrac{\partial \mathcal{E}_{B}}{\partial y^{A}})\circ\chi _{t,1} \diff t -  y_{~\mu }^{A} \int_0^1 t \dfrac{\partial \mathcal{E}_{B}}{\partial y_{\mu~}^{A}}\circ\chi _{t,1} \diff t - y_{~\mu \nu }^{A} \int_0^1 t \dfrac{\partial \mathcal{E}_{B}}{\partial y_{\mu \nu ~}^{A}}\circ\chi _{t,1}\diff t.
\end{align*}%

Similar computations yield:
\begin{align*}
  \diff_{\mu}\left(\dfrac{\partial \mathcal{L}_{1}}{\partial y_{~\mu }^{B}}\right) & =y_{~\mu }^{A}\int_0^1 t\dfrac{\partial \mathcal{E}_{A}}{\partial y_{\mu }^{B}}\circ \chi_{t,1}\diff t + y^{A}\int_0^1 t \left[\diff_{\mu}\left(\dfrac{\partial \mathcal{E}_{A}}{\partial y_{\mu }^{B}}\right)\right]\circ \chi _{t,1}\diff t; \\
  \diff_{\mu}\diff_{\nu}\left(\dfrac{\partial \mathcal{L}_{1}}{\partial y_{~\mu \nu }^{B}}\right)& =y_{~\mu \nu }^{A}\int_0^1 t \dfrac{\partial \mathcal{E}_{A}}{\partial y_{\mu \nu}^{B}}\circ \chi _{t,1}\diff t + 2 y_{~\mu }^{A} \int_0^1 t\left[\diff_{\nu}\left(\dfrac{\partial \mathcal{E}_{A}}{\partial y_{\mu \nu }^{B}}\right)\right]\circ \chi _{t,1}\diff t \\ &+ y^{A} \int_0^1 t \left[\diff_{\mu}\diff_{\nu}\left(\dfrac{\partial \mathcal{E}_{A}}{\partial y_{\mu \nu }^{B}}\right)\right]\circ\chi _{t,1}\diff t.
\end{align*}%
Summing up, we immediately find : 
\begin{align*}
  \dfrac{\delta \mathcal{L}_{1}}{\delta y^{B}} & \coloneqq \dfrac{\partial \mathcal{L}_{1}}{\partial y^{B}} - \diff_{\mu}\left(\dfrac{\partial \mathcal{L}_{1}}{\partial y_{~\mu }^{B}}\right) + \diff_{\mu}\diff_{\nu}\left(\dfrac{\partial \mathcal{L}_{1}}{\partial y_{~\mu \nu }^{B}}\right)\\
  &= \mathcal{E}_{B}-\int_0^1 t[y^{A}(H_{BA}\circ \chi _{t,1})+y_{~\mu}^{A}(H_{BA}^{\mu }\circ \chi _{t,1})+y_{~\mu \nu }^{A}(H_{BA}^{\mu\nu}\circ\chi _{t,1})]\diff t,\\
  &= \mathcal{E}_{B}-\int_0^1 [y^{A}H_{BA}+y_{~\mu}^{A}H_{BA}^{\mu } + y_{~\mu \nu }^{A}H_{BA}^{\mu\nu}]\circ \chi _{t,1} \diff t.
\end{align*}%
where $H_{AB}^{\mu \nu },~H_{~AB}^{\nu },$ $H_{AB}$ are the Helmholtz expressions (\ref{H1})-(\ref{H3}).

Since, by hypothesis, $\mathcal{E}_{\left( 1\right) }$ is locally variational with respect to the variables $y^A$, we get that $H_{AB}=0, H_{AB}^{\mu }=0, H_{AB}^{\mu\nu}=0$, which then implies :
\begin{equation*}
\dfrac{\delta \mathcal{L}_{1}}{\delta y^{B}}=\mathcal{E}_{B},
\end{equation*}%
that is, $\lambda_{1}$ is a Lagrangian for $\mathcal{E}_{\left( 1\right) },$ as claimed.

\section{Classification of pure distortion Lagrangians} \label{appx:classif}

Assume $\dim M=4.$ According to Theorem \ref{lem:degrees_lambda}, natural Lagrangians $\lambda\in \Omega_{4,M}\left(J^{r}T_{2}^{1}M\right) $ to be constructed from the distortion tensor $\tens{L}$ only (with no contribution of the metric) are differential 4-forms that are polynomial and of order at most one in $\tens{L}$. More precisely, these must belong to one of the following classes :

\begin{enumerate}
\item Purely algebraic in $L$ $(d_{0}=4,d_{1}=0)$: 
\begin{eqnarray}
\lambda_{L_{\cdot}L_{\cdot}L_{\cdot}L_{\cdot}} = L_{~\mu_{\sigma \left(1\right)}\tau_{1}}^{\mu_{1}}L_{~\mu_{\sigma \left(2\right)}\tau_{2}}^{\mu_{2}}L_{~\mu_{\sigma \left(3\right)}\tau_{3}}^{\mu_{3}}L_{~\mu_{\sigma \left(4\right)}\tau_{4}}^{\mu_{4}}\diff x^{\tau_{1}}\wedge \diff x^{\tau_{2}}\wedge \diff x^{\tau_{3}}\wedge \diff x^{\tau_{4}};
\label{lambda_1} \\
\lambda_{L_{\cdot\cdot}LL_{\cdot}L_{\cdot}} = L_{~\tau_{1}\tau_{2}}^{\mu_{1}}L_{~\mu_{\sigma \left(1\right)}\mu_{\sigma \left(2\right)}}^{\mu_{2}}L_{~\mu_{\sigma \left(3\right)}\tau_{3}}^{\mu_{3}}L_{~\mu_{\sigma \left(4\right)}\tau_{4}}^{\mu_{4}}\diff x^{\tau_{1}}\wedge \diff x^{\tau_{2}}\wedge \diff x^{\tau_{3}}\wedge \diff x^{\tau_{4}};
\label{lambda_2} \\
\lambda_{L_{\cdot\cdot}L_{\cdot\cdot}LL} = L_{~\tau_{1}\tau_{2}}^{\mu_{1}}L_{~\tau_{3}\tau_{4}}^{\mu_{2}}L_{~\mu_{\sigma \left(1\right)}\mu_{\sigma \left(2\right)}}^{\mu_{3}}L_{~\mu_{\sigma \left(3\right)}\mu_{\sigma \left(4\right)}}^{\mu_{4}}\diff x^{\tau_{1}}\wedge \diff x^{\tau_{2}}\wedge \diff x^{\tau_{3}}\wedge \diff x^{\tau_{4}},  \label{lambda_3}
\end{eqnarray}%
and all possibilities obtained by interchanging indices $\mu_{\sigma \left(i\right)}$ with $\tau_{i}$ in \eqref{lambda_1}--\eqref{lambda_3}. In the above expressions, $\sigma \in S_{4}$ is an arbitrary permutation and the number of dots next to each $L$ in the denomination of $\lambda$ indicates the number of subscripts $\tau_{i}$ of the respective factor that are summed with one of the differentials $\diff x^{\tau_{i}}$.

\item First order in $L$ : These can be quadratic in $\partial L$ ($d_{0}=0,d_{1}=2$), or quadratic in $L$ and linear in $\partial L$ ($d_{0}=2,d_{1}=1$). In both cases, these must be obtained by (horizontal) exterior differentiation $\diff_H = h \circ \diff$ and wedge products from lower rank differential forms which are purely algebraic in $L$.
\end{enumerate}

In order to more easily keep track of all the independent terms, instead of interchanging $\mu_{\sigma \left(i\right)}$ and $\tau_{i}$ in \eqref{lambda_1}--\eqref{lambda_3}, we will split the distortion tensor $\tens{L}$ into its symmetric $\tens{Q}$ and antisymmetric $\tens{T}$ parts\footnote{Pay attention that, here, since $\tens{Q}$ (resp. $\tens{T}$) denotes the symmetric (resp. antisymmetric) part of the distortion tensor $\tens{L}$, it should not be mistaken with the non-metricity (resp. torsion) tensor.} :
\begin{equation}
\ud{L}{\tau}{\beta\gamma} = \ud{Q}{\tau}{(\beta\gamma)} + \ud{T}{\tau}{[\beta\gamma]}.
\end{equation}%
All Lagrangians (\ref{lambda_1})--(\ref{lambda_3}), can be eventually expressed as sums of polynomials in $\tens{Q}$ and $\tens{T}$ -- where some of these polynomials (\emph{e.g.}, all $Q_{\cdot\cdot}$-terms) will vanish due to the antisymmetry of the wedge product.

\subsection{Lower rank, algebraic invariants}\label{sec:lowrankalg}

A preliminary step is to classify lower rank invariant $k$-forms $\rho =\rho_{\tau_{1}\cdots\tau_{k}}(T,Q)\diff x^{\tau_{1}}\wedge \cdots\wedge \diff x^{\tau_{k}}$ ($k\in \left\{ 1,2,3\right\} $) which are algebraic in $\ud{T}{\tau}{\beta\gamma}$ and $\ud{Q}{\tau}{\beta\gamma}$. These are, as mentioned above, polynomial in $\tens{T}$ and $\tens{Q}$; moreover, invariance to arbitrary natural coordinate changes implies that the total degree of the polynomial must be precisely $k.$

\begin{enumerate}
\item[a.] \textbf{1-forms.} For $k=1,$ we get two independent forms : 
\begin{equation}
\tensi{\alpha}_{T}=T_{~\mu \tau}^{\mu}\diff x^{\tau},~\ \tensi{\alpha}_{Q}=Q_{~\mu \tau}^{\mu}\diff x^{\tau}
\end{equation}%
(here we haven't used any dots next to $T$ or $Q$ in the left hand sides as there is no risk of confusion).

\item[b.] \textbf{2-forms.} We easily find four independent such forms :
\begin{equation}
\begin{array}{cc}
\tensi{\alpha}_{T}\wedge \tensi{\alpha}_{Q}, & \tensi{\beta}_{Q_{\cdot}T_{\cdot}}=Q_{~\mu_{2}\tau_{1}}^{\mu_{1}}T_{~\mu_{1}\tau_{2}}^{\mu_{2}}\diff x^{\tau_{1}}\wedge \diff x^{\tau_{2}}, \\ 
\tensi{\beta}_{TT_{\cdot\cdot}}=T_{~\mu_{1}\mu_{2}}^{\mu_{1}}T_{~\tau_{1}\tau_{2}}^{\mu_{2}}\diff x^{\tau_{1}}\wedge \diff x^{\tau_{2}}, & \tensi{\beta}_{QT_{\cdot\cdot}}=Q_{~\mu_{1}\mu_{2}}^{\mu_{1}}T_{~\tau_{1}\tau_{2}}^{\mu_{2}}\diff x^{\tau_{1}}\wedge \diff x^{\tau_{2}}.%
\end{array}
\label{2-forms}
\end{equation}

\item[c.] \textbf{Forms of rank 3.} These must belong to one of the following classes :
\begin{eqnarray}
\tensi{\gamma}_{A_{\cdot}B_{\cdot}C_{\cdot}} = A_{~\mu_{\sigma\left(1\right)}\tau_{1}}^{\mu_{1}}B_{~\mu_{\sigma \left(2\right)}\tau_{2}}^{\mu_{2}}C_{~\mu_{\sigma \left(3\right)}\tau_{3}}^{\mu_{3}}\diff x^{\tau_{1}}\wedge \diff x^{\tau_{2}}\wedge \diff x^{\tau_{3}}, \\
\tensi{\gamma}_{A_{\cdot\cdot}BC_{\cdot}} = A_{~\tau_{1}\tau_{2}}^{\mu_{1}}B_{~\mu_{\sigma \left(1\right)}\mu_{\sigma \left(2\right)}}^{\mu_{2}}C_{~\mu_{\sigma \left(3\right)}\tau_{3}}^{\mu_{3}}\diff x^{\tau_{1}}\wedge \diff x^{\tau_{2}}\wedge \diff x^{\tau_{3}},
\end{eqnarray}%
where $A,B,C\in \left\{ Q,T\right\} $ and $\sigma \in S_{3}.$ By direct inspection, we find 18 independent invariants:

$\checkmark$ 6 wedge products :
\begin{equation}
\tensi{\alpha}_{T}\wedge \tensi{\beta}_{Q_{\cdot}T_{\cdot}},
\tensi{\alpha}_{T}\wedge \tensi{\beta}_{TT_{\cdot\cdot}},
\tensi{\alpha}_{T}\wedge \tensi{\beta}_{QT_{\cdot\cdot}},
\tensi{\alpha}_{Q}\wedge \tensi{\beta}_{Q_{\cdot}T_{\cdot}},
\tensi{\alpha}_{Q}\wedge \tensi{\beta}_{TT_{\cdot\cdot}},
\tensi{\alpha}_{Q}\wedge \tensi{\beta}_{QT_{\cdot\cdot}};
\end{equation}

$\checkmark$ 4 indecomposable 3-forms $\tensi{\gamma}_{A_{\cdot}B_{\cdot}C_{\cdot}}$ :
\begin{eqnarray}
\tensi{\gamma}_{T_{\cdot}T_{\cdot}T_{\cdot}} = T_{~\mu_{2}\tau_{1}}^{\mu_{1}}T_{~\mu_{3}\tau_{2}}^{\mu_{2}}T_{~\mu_{1}\tau_{3}}^{\mu_{3}}\diff x^{\tau_{1}}\wedge \diff x^{\tau_{2}}\wedge \diff x^{\tau_{3}}, \\
\tensi{\gamma}_{T_{\cdot}T_{\cdot}Q_{\cdot}} = T_{~\mu_{2}\tau_{1}}^{\mu_{1}}T_{~\mu_{3}\tau_{2}}^{\mu_{2}}Q_{~\mu_{1}\tau_{3}}^{\mu_{3}}\diff x^{\tau_{1}}\wedge \diff x^{\tau_{2}}\wedge \diff x^{\tau_{3}}, \\
\tensi{\gamma}_{T_{\cdot}Q_{\cdot}Q_{\cdot}} = T_{~\mu_{2}\tau_{1}}^{\mu_{1}}Q_{~\mu_{3}\tau_{2}}^{\mu_{2}}Q_{~\mu_{1}\tau_{3}}^{\mu_{3}}\diff x^{\tau_{1}}\wedge \diff x^{\tau_{2}}\wedge \diff x^{\tau_{3}}, \\
\tensi{\gamma}_{Q_{\cdot}Q_{\cdot}Q_{\cdot}} = Q_{~\mu_{2}\tau_{1}}^{\mu_{1}}Q_{~\mu_{3}\tau_{2}}^{\mu_{2}}Q_{~\mu_{1}\tau_{3}}^{\mu_{3}}\diff x^{\tau_{1}}\wedge \diff x^{\tau_{2}}\wedge \diff x^{\tau_{3}},
\end{eqnarray}

$\checkmark$ 8 indecomposable 3-forms $\tensi{\gamma}_{A_{\cdot \cdot}BC_{\cdot}}:$ 
\begin{eqnarray}
\tensi{\gamma}_{T_{\cdot\cdot}TT_{\cdot}} = T_{~\tau_{1}\tau_{2}}^{\mu_{1}}T_{~\mu_{2}\mu_{3}}^{\mu_{2}}T_{~\mu_{1}\tau_{3}}^{\mu_{3}}\diff x^{\tau_{1}}\wedge \diff x^{\tau_{2}}\wedge \diff x^{\tau_{3}},\\
\tensi{\gamma}_{T_{\cdot\cdot}QT_{\cdot}} = T_{~\tau_{1}\tau_{2}}^{\mu_{1}}Q_{~\mu_{2}\mu_{3}}^{\mu_{2}}T_{~\mu_{1}\tau_{3}}^{\mu_{3}}\diff x^{\tau_{1}}\wedge \diff x^{\tau_{2}}\wedge \diff x^{\tau_{3}},\\
\tensi{\gamma}_{T_{\cdot\cdot}TQ_{\cdot}} = T_{~\tau_{1}\tau_{2}}^{\mu_{1}}T_{~\mu_{2}\mu_{3}}^{\mu_{2}}Q_{~\mu_{1}\tau_{3}}^{\mu_{3}}\diff x^{\tau_{1}}\wedge \diff x^{\tau_{2}}\wedge \diff x^{\tau_{3}},\\
\tensi{\gamma}_{T_{\cdot\cdot}QQ_{\cdot}} = T_{~\tau_{1}\tau_{2}}^{\mu_{1}}Q_{~\mu_{2}\mu_{3}}^{\mu_{2}}Q_{~\mu_{1}\tau_{3}}^{\mu_{3}}\diff x^{\tau_{1}}\wedge \diff x^{\tau_{2}}\wedge \diff x^{\tau_{3}},
\end{eqnarray}
\begin{eqnarray}
\tensi{\tilde{\gamma}}_{T_{\cdot\cdot}TT_{\cdot}}=T_{~\tau_{1}\tau_{2}}^{\mu_{1}}~T_{~\mu_{1}\mu_{3}}^{\mu_{2}}~T_{~\mu_{2}\tau_{3}}^{\mu_{3}}\diff x^{\tau_{1}}\wedge \diff x^{\tau_{2}}\wedge \diff x^{\tau_{3}},\\
\tensi{\tilde{\gamma}}_{T_{\cdot\cdot}QT_{\cdot}}=T_{~\tau_{1}\tau_{2}}^{\mu_{1}}~Q_{~\mu_{1}\mu_{3}}^{\mu_{2}}~T_{~\mu_{2}\tau_{3}}^{\mu_{3}}\diff x^{\tau_{1}}\wedge \diff x^{\tau_{2}}\wedge \diff x^{\tau_{3}}, \\
\tensi{\tilde{\gamma}}_{T_{\cdot\cdot}TQ_{\cdot}}=T_{~\tau_{1}\tau_{2}}^{\mu_{1}}T_{~\mu_{1}\mu_{3}}^{\mu_{2}}Q_{~\mu_{2}\tau_{3}}^{\mu_{3}}\diff x^{\tau_{1}}\wedge \diff x^{\tau_{2}}\wedge \diff x^{\tau_{3}}, \\
\tensi{\tilde{\gamma}}_{T_{\cdot\cdot}QQ_{\cdot}}=T_{~\tau_{1}\tau_{2}}^{\mu_{1}}Q_{~\mu_{1}\mu_{3}}^{\mu_{2}}Q_{~\mu_{2}\tau_{3}}^{\mu_{3}}\diff x^{\tau_{1}}\wedge \diff x^{\tau_{2}}\wedge \diff x^{\tau_{3}}.
\end{eqnarray}
\end{enumerate}

\bigskip

\subsection{Zeroth order (algebraic) Lagrangians in $\tens{L}$}

These must be expressed as homogeneous polynomials of total degree 4 in $\ud{T}{\tau}{\beta \gamma}$ and $\ud{Q}{\tau}{\beta \gamma}.$ There are 65 such independent invariants, as follows :

$\checkmark$ 33 wedge products of lower rank differential forms :
\begin{eqnarray}
&&\tensi{\alpha}_{T}\wedge \tensi{\gamma}_{T_{\cdot}T_{\cdot}T_{\cdot}},\cdots,~\tensi{\alpha}_{T}\wedge \tensi{\tilde{\gamma}}_{T_{\cdot\cdot}QQ_{\cdot}},~\tensi{\alpha}_{Q}\wedge\tensi{\gamma}_{T_{\cdot}T_{\cdot}T_{\cdot}},\cdots,~\tensi{\alpha}_{T}\wedge\tensi{\tilde{\gamma}}_{T_{\cdot\cdot}QQ_{\cdot}},\\
&&\tensi{\alpha}_{T}\wedge \tensi{\alpha}_{Q}\wedge \tensi{\beta}_{Q_{\cdot}T_{\cdot}},~\ \tensi{\alpha}_{T}\wedge \tensi{\alpha}_{Q}\wedge\tensi{\beta}_{TT_{\cdot\cdot}},~\ \tensi{\alpha}_{T}\wedge \tensi{\alpha}_{Q}\wedge \tensi{\beta}_{QT_{\cdot\cdot}}, \\
&&\tensi{\beta}_{Q_{\cdot}T_{\cdot}}\wedge \tensi{\beta}_{Q_{\cdot}T_{\cdot}},~\ \tensi{\beta}_{Q_{\cdot}T_{\cdot}}\wedge\tensi{\beta}_{TT_{\cdot\cdot}},\ \tensi{\beta}_{Q_{\cdot}T_{\cdot}}\wedge \tensi{\beta}_{QT_{\cdot\cdot}},\\
&&\tensi{\beta}_{TT_{\cdot\cdot}}\wedge \tensi{\beta}_{TT_{\cdot\cdot}},\ \tensi{\beta}_{TT_{\cdot\cdot}}\wedge \tensi{\beta}_{QT_{\cdot\cdot}},~\tensi{\beta}_{QT_{\cdot\cdot}}\wedge \tensi{\beta}_{QT_{\cdot\cdot}};
\end{eqnarray}%
and 32 indecomposable 4-forms, grouped as follows :

$\checkmark$ 5 independent forms of type%
\begin{equation}
\tensi{\lambda}_{T_{\cdot\cdot}T_{\cdot\cdot}AB}=T_{~\tau_{1}\tau_{2}}^{\mu_{1}}T_{~\tau_{3}\tau_{4}}^{\mu_{2}}A_{~\mu_{\sigma \left(1\right)}\mu_{\sigma \left(2\right)}}^{\mu_{3}}B_{~\mu_{\sigma\left(3\right)}\mu_{\sigma\left(4\right)}}^{\mu_{4}}\diff x^{\tau_{1}}\wedge \diff x^{\tau_{2}}\wedge \diff x^{\tau_{3}}\wedge \diff x^{\tau_{4}},
\end{equation}%
where $A,B\in \left\{T,Q\right\} ,$ found for $\sigma =\left(1,2,3,4\right) $ and $\sigma =\left(1,4,2,3\right)$ :
\begin{eqnarray}
\tensi{\lambda}_{T_{\cdot\cdot}T_{\cdot\cdot}QT} = T_{~\tau_{1}\tau_{2}}^{\mu_{1}}T_{~\tau_{3}\tau_{4}}^{\mu_{2}}Q_{~\mu_{1}\mu_{2}}^{\mu_{3}}T_{~\mu_{3}\mu_{4}}^{\mu_{4}}\diff x^{\tau_{1}}\wedge \diff x^{\tau_{2}}\wedge \diff x^{\tau_{3}}\wedge \diff x^{\tau_{4}}, \\
\tensi{\lambda}_{T_{\cdot\cdot}T_{\cdot\cdot}QQ} = T_{~\tau_{1}\tau_{2}}^{\mu_{1}}T_{~\tau_{3}\tau_{4}}^{\mu_{2}}Q_{~\mu_{1}\mu_{2}}^{\mu_{3}}Q_{~\mu_{3}\mu_{4}}^{\mu_{4}}\diff x^{\tau_{1}}\wedge \diff x^{\tau_{2}}\wedge \diff x^{\tau_{3}}\wedge \diff x^{\tau_{4}},
\end{eqnarray}
\begin{eqnarray}
\tensi{\tilde{\lambda}}_{T_{\cdot\cdot}T_{\cdot\cdot}TT} = T_{~\tau_{1}\tau_{2}}^{\mu_{1}}T_{~\tau_{3}\tau_{4}}^{\mu_{2}}T_{~\mu_{1}\mu_{4}}^{\mu_{3}}T_{~\mu_{2}\mu_{3}}^{\mu_{4}}\diff x^{\tau_{1}}\wedge \diff x^{\tau_{2}}\wedge \diff x^{\tau_{3}}\wedge \diff x^{\tau_{4}}, \\
\tensi{\tilde{\lambda}}_{T_{\cdot\cdot}T_{\cdot\cdot}TQ} = T_{~\tau_{1}\tau_{2}}^{\mu_{1}}T_{~\tau_{3}\tau_{4}}^{\mu_{2}}T_{~\mu_{1}\mu_{4}}^{\mu_{3}}Q_{~\mu_{2}\mu_{3}}^{\mu_{4}}\diff x^{\tau_{1}}\wedge \diff x^{\tau_{2}}\wedge \diff x^{\tau_{3}}\wedge \diff x^{\tau_{4}}, \\
\tensi{\tilde{\lambda}}_{T_{\cdot\cdot}T_{\cdot\cdot}QQ} = T_{~\tau_{1}\tau_{2}}^{\mu_{1}}T_{~\tau_{3}\tau_{4}}^{\mu_{2}}Q_{~\mu_{1}\mu_{4}}^{\mu_{3}}Q_{~\mu_{2}\mu_{3}}^{\mu_{4}}\diff x^{\tau_{1}}\wedge \diff x^{\tau_{2}}\wedge \diff x^{\tau_{3}}\wedge \diff x^{\tau_{4}}.
\end{eqnarray}%
(the forms $\tensi{\lambda}_{T_{\cdot\cdot}T_{\cdot\cdot}TT}$, $\lambda_{T_{\cdot\cdot}T_{\cdot\cdot}TQ}$ vanish for symmetry reasons).

$\checkmark$ 24 independent forms of type :
\begin{equation}
\tensi{\lambda}_{T_{\cdot\cdot}AB_{\cdot}C_{\cdot}}=T_{~\tau_{1}\tau_{2}}^{\mu_{1}}A_{~\mu_{\sigma \left(1\right)}\mu_{\sigma \left(2\right)}}^{\mu_{2}}B_{~\mu_{\sigma \left(3\right)}\tau_{3}}^{\mu_{3}}C_{~\mu_{\sigma\left(4\right)}\tau_{4}}^{\mu_{4}}\diff x^{\tau_{1}}\wedge \diff x^{\tau_{2}}\wedge \diff x^{\tau_{3}}\wedge \diff x^{\tau_{4}};
\end{equation}%
with $A,B,C\in\left\{T,Q\right\}$. These can be sorted in three groups (giving 8 independent forms each) corresponding to $\sigma \in \left\{ \left(1,3,4,2\right) ,\left(2,3,4,1\right) ,\left(3,4,1,2\right) \right\}$ : 
\begin{eqnarray}
\tensi{\lambda}_{T_{\cdot\cdot}AB_{\cdot}C_{\cdot}} = T_{~\tau_{1}\tau_{2}}^{\mu_{1}}A_{~\mu_{1}\mu_{3}}^{\mu_{2}}B_{~\mu_{4}\tau_{3}}^{\mu_{3}}C_{~\mu_{2}\tau_{4}}^{\mu_{4}}\diff x^{\tau_{1}}\wedge \diff x^{\tau_{2}}\wedge \diff x^{\tau_{3}}\wedge \diff x^{\tau_{4}}, \\
\tensi{\tilde{\lambda}}_{T_{\cdot\cdot}AB_{\cdot}C_{\cdot}} = T_{~\tau_{1}\tau_{2}}^{\mu_{1}}A_{~\mu_{2}\mu_{3}}^{\mu_{2}}B_{~\mu_{4}\tau_{3}}^{\mu_{3}}C_{~\mu_{1}\tau_{4}}^{\mu_{4}}\diff x^{\tau_{1}}\wedge \diff x^{\tau_{2}}\wedge \diff x^{\tau_{3}}\wedge \diff x^{\tau_{4}}, \\
\tensi{\hat{\lambda}}_{T_{\cdot\cdot}AB_{\cdot}C_{\cdot}} = T_{~\tau_{1}\tau_{2}}^{\mu_{1}}A_{~\mu_{3}\mu_{4}}^{\mu_{2}}B_{~\mu_{1}\tau_{3}}^{\mu_{3}}C_{~\mu_{2}\tau_{4}}^{\mu_{4}}\diff x^{\tau_{1}}\wedge \diff x^{\tau_{2}}\wedge \diff x^{\tau_{3}}\wedge \diff x^{\tau_{4}},
\end{eqnarray}

$\checkmark$ 3 independent forms of type%
\begin{equation}
\tensi{\lambda}_{A_{\cdot}B_{\cdot}C_{\cdot}D_{\cdot}}=A_{~\mu_{\sigma \left(1\right)}\tau_{1}}^{\mu_{1}}B_{~\mu_{\sigma \left(2\right)}\tau_{2}}^{\mu_{2}}C_{~\mu_{\sigma \left(3\right)}\tau_{3}}^{\mu_{3}}D_{~\mu_{\sigma\left(4\right)}\tau_{4}}^{\mu_{4}}\diff x^{\tau_{1}}\wedge \diff x^{\tau_{2}}\wedge \diff x^{\tau_{3}}\wedge \diff x^{\tau_{4}};
\end{equation}%
indecomposable such forms are generated by \textit{derangements} $\sigma \in S_{4}$ (i.e. $\sigma \left(i\right) \not=i,$ $\forall i$); all these permutations lead, up to a sign, to the same results :
\begin{eqnarray}
\tensi{\lambda}_{T_{\cdot}T_{\cdot}T_{\cdot}Q_{\cdot}} = T_{~\mu_{2}\tau_{1}}^{\mu_{1}}T_{~\mu_{3}\tau_{2}}^{\mu_{2}}T_{~\mu_{4}\tau_{3}}^{\mu_{3}}Q_{~\mu_{1}\tau_{4}}^{\mu_{4}}\diff x^{\tau_{1}}\wedge \diff x^{\tau_{2}}\wedge \diff x^{\tau_{3}}\wedge \diff x^{\tau_{4}}, \\
\tensi{\lambda}_{T_{\cdot}T_{\cdot}Q_{\cdot}Q_{\cdot}} = T_{~\mu_{2}\tau_{1}}^{\mu_{1}}T_{~\mu_{3}\tau_{2}}^{\mu_{2}}Q_{~\mu_{4}\tau_{3}}^{\mu_{3}}Q_{~\mu_{1}\tau_{4}}^{\mu_{4}}\diff x^{\tau_{1}}\wedge \diff x^{\tau_{2}}\wedge \diff x^{\tau_{3}}\wedge \diff x^{\tau_{4}}, \\
\tensi{\lambda}_{T_{\cdot}Q_{\cdot}Q_{\cdot}Q_{\cdot}} = T_{~\mu_{2}\tau_{1}}^{\mu_{1}}T_{~\mu_{3}\tau_{2}}^{\mu_{2}}T_{~\mu_{4}\tau_{3}}^{\mu_{3}}Q_{~\mu_{1}\tau_{4}}^{\mu_{4}}\diff x^{\tau_{1}}\wedge \diff x^{\tau_{2}}\wedge \diff x^{\tau_{3}}\wedge \diff x^{\tau_{4}},
\end{eqnarray}%
the forms $\tensi{\lambda}_{T_{\cdot}T_{\cdot}T_{\cdot}T_{\cdot}}$ and $\tensi{\lambda}_{Q_{\cdot}Q_{\cdot}Q_{\cdot}Q_{\cdot}}$ obtained this way vanish; also, the symmetry of the indices implies, \emph{e.g.}, $\tensi{\lambda}_{T_.T_.Q_.T_.} \propto \tensi{\lambda}_{T_.T_.T_.Q_.}$.

\subsection{First order Lagrangians}

We will now discuss natural Lagrangians on $J^1T^1_2M$ that are first order in $\tens{L}$ (\emph{i.e.} that depend on $\partial_{\sigma} \ud{L}{\rho}{\mu\nu}$). These must be obtained from the lower rank algebraic invariants from section \ref{sec:lowrankalg} via (horizontal) exterior derivative and wedge products.

We then find 29 independent such terms :

\begin{itemize}
\item 3 independent Lagrangians quadratic in $\partial L$ :
\begin{equation}
\diff_H\tensi{\alpha}_{T}\wedge \diff_H\tensi{\alpha}_{T},~\ \diff_H\tensi{\alpha}_{T}\wedge \diff_H\tensi{\alpha}_{Q},~\ \diff_H\tensi{\alpha}_{Q}\wedge \diff_H\tensi{\alpha}_{Q}.
\end{equation}

These are, yet, all total divergences as, \emph{e.g.}, $\diff_H\tensi{\alpha}_{T}\wedge \diff_H\tensi{\alpha}_{T} = \diff_H\left(\tensi{\alpha}_{T}\wedge \diff_H\tensi{\alpha}_{T}\right) .$ Thus, they produce trivial Euler-Lagrange equations.

\item 14 Lagrangians linear in $\partial L$ involving $\tensi{\alpha}$ and $\tensi{\beta}$ terms  :
\begin{eqnarray}
&&\diff_H\tensi{\alpha}_{T}\wedge \tensi{\alpha}_{T}\wedge \tensi{\alpha}_{Q}, \\
&&\diff_H\tensi{\alpha}_{Q}\wedge \tensi{\alpha}_{T}\wedge \tensi{\alpha}_{Q}, \\
&&\diff_H\tensi{\alpha}_{T}\wedge \tensi{\beta}_{Q_{\cdot}T_{\cdot}},~\ \tensi{\alpha}_{T}\wedge \diff_H\tensi{\beta}_{Q_{\cdot}T_{\cdot}}, \label{eq:pair1} \\
&&\diff_H\tensi{\alpha}_{Q}\wedge \tensi{\beta}_{Q_{\cdot}T_{\cdot}},~\ \tensi{\alpha}_{Q}\wedge \diff_H\tensi{\beta}_{Q_{\cdot}T_{\cdot}}, \\
&&\diff_H\tensi{\alpha}_{T}\wedge \tensi{\beta}_{TT_{\cdot\cdot}},~\ \tensi{\alpha}_{T}\wedge \diff_H\tensi{\beta}_{TT_{\cdot\cdot}}, \\
&&\diff_H\tensi{\alpha}_{Q}\wedge \tensi{\beta}_{TT_{\cdot\cdot}},~\ \tensi{\alpha}_{Q}\wedge \diff_H\tensi{\beta}_{TT_{\cdot\cdot}}, \\
&&\diff_H\tensi{\alpha}_{T}\wedge \tensi{\beta}_{QT_{\cdot\cdot}},~\ \tensi{\alpha}_{T}\wedge \diff_H\tensi{\beta}_{QT_{\cdot\cdot}}, \\
&&\diff_H\tensi{\alpha}_{Q}\wedge \tensi{\beta}_{QT_{\cdot\cdot}},~\ \tensi{\alpha}_{Q}\wedge \diff_H\tensi{\beta}_{QT_{\cdot\cdot}}.\label{eq:pair6}
\end{eqnarray}

Of these, each pair \eqref{eq:pair1}--\eqref{eq:pair6} above is made by Lagrangians which differ by a total derivative\footnote{Indeed, since the $\tensi{\alpha}_i$ terms are of rank 1, $\diff_H\left(\tensi{\alpha}_i\wedge\tensi{\beta}_j\right) = \diff_H\tensi{\alpha}_i\wedge\tensi{\beta}_j - \tensi{\alpha}_i\wedge\diff_H\tensi{\beta}_j$.}, hence, they produce the same Euler-Lagrange equations. We obtain, thus, 8 possibilities of producing different Euler-Lagrange equations. These correspond \emph{e.g.} to the wedge products of the two first order rank 2 terms in $\tens{L}$ ($\diff_H\tensi{\alpha}_T$ and $\diff_H\tensi{\alpha}_Q$) with the four algebraic rank 2 terms in $\tens{L}$ (\emph{i.e.} $\tensi{\alpha}_T\wedge\tensi{\alpha}_Q$, $\tensi{\beta}_{Q_.T_.}$, $\tensi{\beta}_{T{T}_{..}}$ and $\tensi{\beta}_{QT_{..}}$).


\item 12 Lagrangians linear in $\partial L$ of type $\diff_H \tensi{\gamma}_{ABC}$.

These are all obtained by horizontal exterior differentiation; hence, their Euler-Lagrange expressions identically vanish.
\end{itemize}

To conclude, on 4-dimensional manifolds, one can build 65 algebraic invariants, but only 8 nontrivial first order, independent Lagrangians, using the distortion tensor alone.

Interesting remark : there are only one algebraic invariant Lagrangian $\tensi{\alpha}_{Q}\wedge \tensi{\gamma}_{Q_{\cdot}Q_{\cdot}Q_{\cdot}}$ and no first order, nontrivial Lagrangians which can be built from the symmetric part of $\tens{L}$ only.

\end{document}